%% file: main.tex
\documentclass[11pt]{article}
\usepackage[top=1in,bottom=1in,right=.8in,left=.8in]{geometry}

\usepackage{amssymb,amsmath,amsthm}
\usepackage{tikz}
\usepackage{arydshln}
\usepackage{verbatim}
\usepackage{titlesec}
\usepackage{bbm}
\usepackage{bm}
\usepackage{subfigure}
\usepackage{ stmaryrd }
\usepackage{cleveref}
\usepackage{xcolor}
\usepackage{color}
\usepackage[suppress]{color-edits}
\addauthor{et}{purple}
\addauthor{jg}{blue}

\newcounter{subsubsubsection}[subsubsection]
\renewcommand\thesubsubsubsection{\thesubsubsection.\arabic{subsubsubsection}}

\titleformat{\subsubsubsection}
  {\normalfont\normalsize\bfseries}{\thesubsubsubsection}{1em}{}
  \titlespacing*{\subsubsubsection}
  {0pt}{3.25ex plus 1ex minus .2ex}{1.5ex plus .2ex}

\makeatletter  
\def\toclevel@subsubsubsection{4}
\def\l@subsubsubsection{\@dottedtocline{4}{7em}{4em}}

\makeatother
\newcommand{\headder}[1]{\textbf{#1}}

\crefname{equation}{}{}
\crefrangeformat{equation}{(#3#1#4) --~(#5#2#6)}
\crefmultiformat{equation}{(#2#1#3)}{, (#2#1#3)}{, (#2#1#3)}{, (#2#1#3)}
\crefname{lem}{Lemma}{Lemmas}
\crefname{section}{Section}{Sections}
\crefname{subsubsubsection}{Section}{Sections}
\crefname{rem}{Remark}{Remarks}
\crefname{figure}{Figure}{Figures}
\crefname{table}{Table}{Tables}
\Crefname{lem}{Lemma}{Lemmas}
\crefname{thm}{Theorem}{Theorems}
\Crefname{thm}{Theorem}{Theorems}

\newtheorem{thm}{Theorem}[section]
\newtheorem{assumption}{Assumption}[section]
\newtheorem{claim}{Claim}[section]

\newtheorem{remark}{Remark}[section]

\newtheorem{lem}{Lemma}[section]

\newtheorem{corollary}[thm]{Corollary}

\theoremstyle{definition}
\newtheorem{fact}{Fact}[section]
 
\theoremstyle{definition}
\newtheorem{defn}{Definition}[section]
\renewenvironment{proof}{\noindent {\bf Proof.}}{\qed}

\title{Stability and Learning in Strategic Queuing Systems}
\author{Jason Gaitonde\footnote{Cornell University, Department of Applied Mathematics. Supported by NSF grant CCF-1408673 and ASOFR grant F5684A1. Email: jsg355@cornell.edu.}\, and \'Eva Tardos\footnote{Cornell University, Department of Computer Science. Supported  in part by NSF grant CCF-1408673, CCF-1563714, and ASOFR grant F5684A1. Email: eva.tardos@cornell.edu.}}
\date{February 2020}

\begin{document}

\maketitle
\begin{abstract}
\input{abstract}

\end{abstract}
\input{intro.tex}

\section{Preliminaries}

\headder{Notation.}
In general, random variables will be denoted by capital letters (i.e. $X,Y,Z,\ldots)$, while vectors will generally be bolded (i.e. $\bm{\mu},\bm{\lambda}$, etc). If a random variable $X$ has some distribution $\mathcal{D}$, we write $X\sim \mathcal{D}$. We use the notation $\text{Geom}(p)$ to denote a geometric distribution with parameter $p$, $\text{Bern}(p)$ for a Bernoulli distribution that is $1$ with probability $p$ and $0$ otherwise, and $\text{Bin}(n,p)$ for a binomial distribution with parameters $n$ and $p$. 

We say an event occurs \emph{almost surely} if it has probability $1$. We use standard $O(\cdot), o(\cdot),$ and $\Theta(\cdot)$ notation, where $\tilde{O}(\cdot)$ indicates logarithmic factors are hidden; we will sometimes write $f(n)\asymp g(n)$ if $f(n)=\Theta(g(n))$. We will also consider the following norms: for a positive vector $\bm{\lambda}=(\lambda_1,\ldots,\lambda_n)$, with $\lambda_1\geq\ldots\geq \lambda_n>0$, we define the following two weighted $\ell_p$ norms on $\mathbb{R}^n$:$
    \|\mathbf{x}\|_{\bm{\lambda},1}\triangleq \sum_{i=1}^n \lambda_i \vert x_i\vert$ and $
    \|\mathbf{x}\|_{\bm{\lambda},2}\triangleq \sqrt{\sum_{i=1}^n \lambda_i x_i^2}.$
It is easily seen that for any $\mathbf{x}$, $\|\mathbf{x}\|_{\bm{\lambda},1}\asymp \|\mathbf{x}\|_{\bm{\lambda},2}$ (where the constants depend on $\bm{\lambda}$) via Cauchy-Schwarz, see Lemma \ref{lem:norms}.

\headder{Standard Queuing Model.}
\label{section:standardsystem}
We consider the following discrete-time queuing system illustrated by the figure below, which is a decentralized, competitive version of the model considered by Krishnasamy, et al \cite{krishnasamy2016learning}: there is a system of $n$ queues and $m$ servers. 
During each discrete time step $t=0,1,\ldots$, the following occurs: 
\begin{enumerate}
    \item Each queue $i$ receives a new packet with a fixed, time-independent probability $\lambda_i$. We model this via an independent random variable $B^i_t\sim \text{Bern}(\lambda_i)$. This packet has a timestamp that indicates that it was generated in the current time period. We label queues such that $\lambda_1\geq \ldots\geq \lambda_n>0$, writing $\bm{\lambda}$ for the vector of arrival rates. 
    \item Each queue that currently has an uncompleted packet chooses one server to send their oldest unprocessed packet (in terms of timestamp) to.
    \item Each server $j$ that receives a packet does the following: first, it only considers the packet it receives with the oldest timestamp (breaking ties arbitrarily). It then processes this packet with a fixed, time-independent probability $\mu_j$. We again label servers so that $\mu_1\geq \ldots\geq \mu_{m}\geq 0$, writing $\bm{\mu}$ for the vector of service rates.
    
    \item All unprocessed packets, possibly including the packets that were selected if the corresponding server failed to process it, are then sent back to their respective queues still uncompleted. Queues receive bandit feedback on whether their packet cleared at their chosen server.
\end{enumerate}

\begin{figure}[h]
\begin{center}
\hspace{-5cm}
\includegraphics[width=10.5cm,trim=2 2cm 0 2cm,clip]{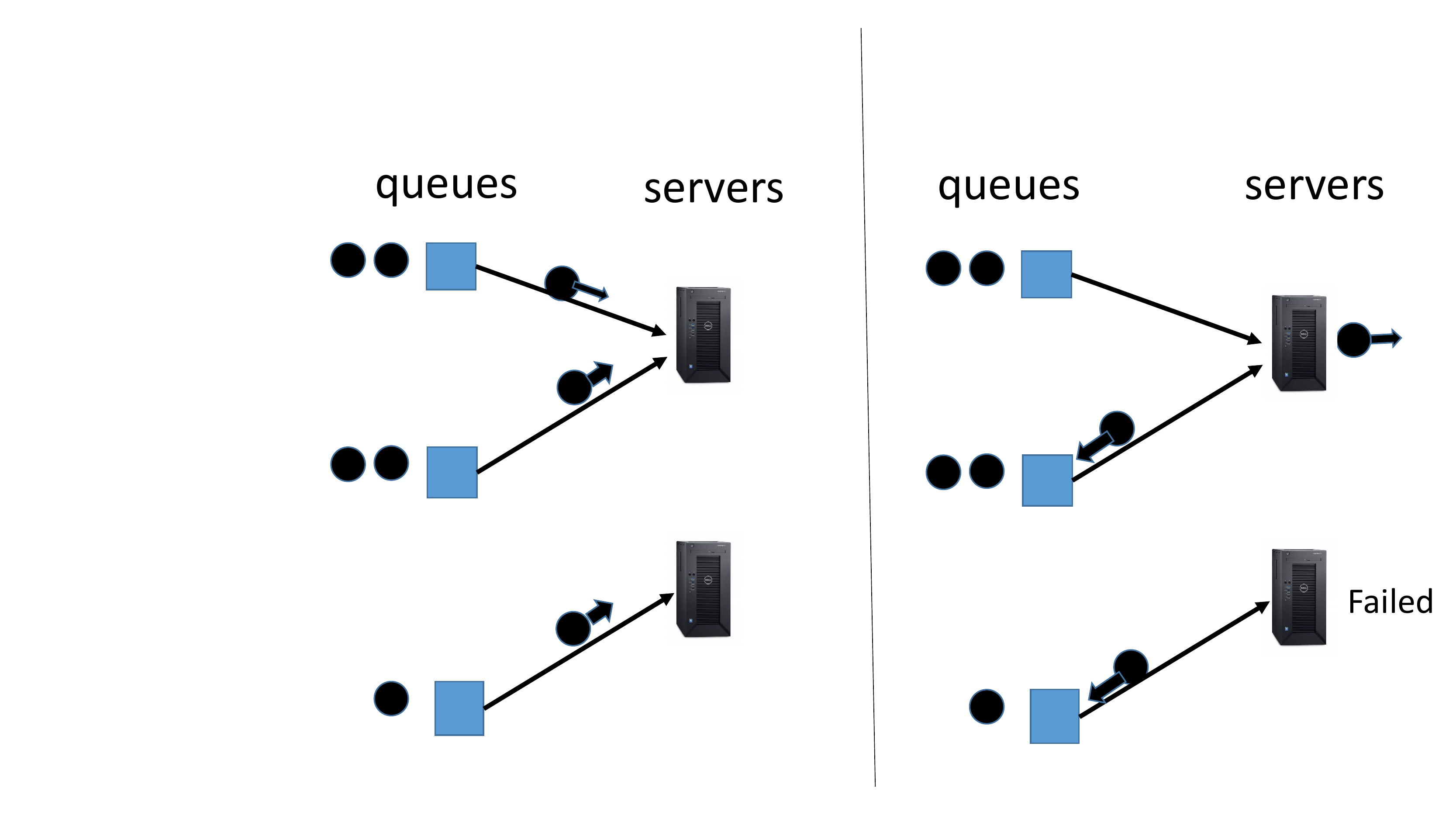}
\end{center}
\caption{Here, three queues compete for two servers. Unserviced packets in each round return to their queue.}
\end{figure}

We write $Q^i_t$ for the number of unprocessed packets of queue $i$ at the beginning of time $t$ (before sampling new packets) and $\mathbf{Q}_t=(Q^1_t,\ldots,Q^n_t)$ for the vector of queue sizes at time $t$. Define $Q_t=\sum_{i=1}^n Q^i_t$ as the total number of unprocessed packets in the system at time $t$. Formally, if $S^i_t$ is the indicator event that queue $i$ clears a packet at time $t$ and $B^i_t$ is again the indicator queue $i$ received a new packet at time $t$, then we have the recurrence as random variables with $Q^i_0=0$ and
\begin{equation}
    Q^i_{t+1}=Q^i_{t}+B^i_t-S^i_t,
\end{equation}
where we note that $S^i_t$ is necessarily $0$ if $Q^i_{t}+B^i_t=0$ (i.e. queue $i$ had no packets and didn't receive a new one in the round, so does not send a packet this time period). This ensures each $Q^i_t$ is integral and nonnegative.
We call the above random process the \textbf{standard model}. We will be interested in the stability of this system in the following sense:

\begin{defn}
The above system is \textbf{strongly stable} under some given dynamics if, for any fixed $r\geq 0$, the random process $Q_t$ satisfies $\mathbb{E}[Q_t^r]\leq C_r$ for some absolute constant $C_r$ that does not depend on $t$.

We say that it is \textbf{almost surely stable} if for any $c>0$, almost surely
\begin{equation*}
    Q_t=o(t^c).
\end{equation*}
That is, the growth of $Q_t$ is almost surely subpolynomial.
\end{defn}

It is not immediately obvious what the relationship is between these stability properties: it turns out that strong stability implies almost sure stability, which we state in Lemma \ref{lem:equivalence} and prove in the Appendix. 

Our main goal is to understand the stability properties of these queuing systems when queues select servers as no-regret learners. To get a baseline measure for when this may be possible, we must first understand when a queuing system is stable under centralized coordination: it turns out that an obvious necessary condition on $\bm{\mu}$ and $\bm{\lambda}$ is also sufficient.

\begin{thm}
\label{thm:feasibility}
Suppose $\bm{\mu}$ and $\bm{\lambda}$ have been preprocessed so that a maximal, equal prefix of $1$'s is deleted from both and both are nonempty and not both identically zero afterwards. Then the above queuing system is strongly stable for some centralized (coordinated) scheduling policy if and only if for all $1\leq k\leq n$,
\begin{equation}
\label{eq:centfeas}
    \sum_{i=1}^k \mu_i > \sum_{i=1}^k \lambda_i.
\end{equation}
When (\ref{eq:centfeas}) holds, we say that the queuing system is \textbf{(centrally) feasible}.
\end{thm}

An instructive example to keep in mind is a single-queue, single-server system. Of course, there is no learning or competition in such a process. If $0<\lambda<\mu\leq 1$, it is well-known that $Q^1_t$ follows a biased random walk on the nonnegative integers towards $0$, and moreover is geometrically ergodic.\footnote{Namely, this random process mixes to a stationary distribution on $\mathbb{N}$ with geometrically decreasing tail probabilities.} This in particular implies strong stability. On the other hand, if $0<\lambda=\mu<1$, say $\mu=\lambda=.5$, then it is well-known that the corresponding unbiased random walk $Q^1_t$ satisfies $\mathbb{E}[Q^1_t]=\Theta(\sqrt{t})$. Therefore, there is a sharp threshold for strong stability. 

We give the full proof of the Theorem in the Appendix and only sketch it here: for necessity, when (\ref{eq:centfeas}) strictly fails, it will be easy to see the queueing system is not stable as the expected total number of packets grows linearly. A slight modification is needed using standard submartingale arguments if instead only equality holds somewhere in (\ref{eq:centfeas}). For sufficiency, we leverage the well-known connection between majorization of vectors with doubly stochastic matrices. By decomposing these matrices into a convex combination of permutation matrices, one obtains a randomized, coordinated matching schedule between queues and servers that ensures at every time step that the probability of clearing a packet strictly exceeds the probability of receiving a new packet. Each queue size thus follows a biased random walk towards $0$, which will ensure stability.

\headder{The Need for Packet Priorities.} 
We will be interested in proving statements of the following form:
\begin{quote}
    Given a queuing system that is centrally feasible even when $\bm{\lambda}$ is scaled up by some explicit constant $c\geq 1$ independent of the parameters of the system (in the sense of Theorem \ref{thm:feasibility}), then a random process where queues are decentralized and strategic under certain conditions remains stable.
\end{quote}
Morally, such a result says that though decentralized, strategic queues cannot coordinate and instead compete for service, if they choose servers according to some reasonable learning algorithm, there only needs to be a small multiplicative factor of slack to keep the queuing system stable. In other words, decentralizing a feasible queuing system and introducing a constant factor of slack will result in a stable system even when queues selfishly compete to clear  their own packets.

To see the necessity of timestamps, consider instead a simpler model where there are no timestamps and priorities, and instead each server uniformly randomly picks which packet to process among those that are sent to it in each step. It is easy to see that 
if a queuing system is feasible even if $\bm{\lambda}$ scaled up by $n$, 
then it will remain a stable queuing system with reasonably strategic queues. Indeed, by this feasibility assumption, 
$\mu_1>n\cdot \lambda_1$, 
so that $\mu_1> \sum_{i=1}^n \lambda_i$. Therefore, if every queue just always sends to the largest server whenever they have a packet, they will succeed in clearing a packet with probability at least $1/n>\lambda_i$, and it is not too difficult to prove that this results in a strongly stable process by comparing to a random walk biased towards the origin.

It is natural to ask if a better factor is attainable in this alternate model, perhaps even a constant. It turns out that in general, a polynomial in $n$ is required:
\begin{thm}
\label{thm:impossibility}
In this alternate model, for large enough $n$, there exists a centrally feasible queuing system with $n$ queues and servers with the following property: the system remains feasible even if $\bm{\lambda}$ is scaled up by $\Omega(n^{1/3})$ and it is possible for all queues to be in a Nash equilibrium\footnote{By this we mean that conditioned on the (randomized) strategies of all other queues in a given time step, each queue sends to a server with highest probability of success.} at each time step (and in particular, satisfy no-regret properties as in Assumption \ref{assumption:learning}), yet the system is not strongly stable.
\end{thm}
While we make little effort to optimize the exponent, this shows that in this model, no sub-polynomial factor is possible in general. The basic reason why this can occur is that low arrival rate queues can saturate the high success rate servers, making it impossible for high arrival rate queues to clear fast enough to offset their higher arrival. In the timestamp model, we will be able to establish constant factor results. The key idea 
is that the priority system, while more difficult to analyze, results in older queues gaining an advantage on young queues causing the young queues to prefer lower quality servers,
so that this situation is impossible. That is, this model implicitly forces fast growing queues to get better service, so long as queues are sufficiently adaptive to take advantage of it.

\section{No-Regret in Queuing Systems.}
Our main result shows that, if the queuing system has enough slack and all queues satisfy an appropriate high-probability no-regret guarantee, then the queuing system is strongly stable. To this end, we make the following feasibility assumption that asserts that a queuing systems with servers scaled down by $1/2$ would remain feasible:

\begin{assumption}[Feasibility]
\label{assumption:feasibility}
There exists $\eta>0$ such that for all $k\in [n]$,
\begin{equation*}
    \frac{1}{2}(1-\eta)\sum_{i=1}^k \mu_i\geq \sum_{i=1}^k \lambda_i.
\end{equation*}
We will usually use $\eta$ to denote the maximum such value that this inequality holds.
\end{assumption}
This assumption stipulates that if the service rates were halved, then the queuing system would still be centrally feasible. The parameter $\eta$ controls the quality of learning required for our results. To establish stability results, we will use the following assumption on no-regret algorithms:
\begin{defn}
Fix some window of length $w$, and for convenience reindex time so that the beginning of this window is at $t=1$. Let $S^{i,j}_t$ be the indicator variable that queue $i$ \emph{would have} succeeded in clearing a packet at server $j$ at time $t$ (had she sent there), and let $\sigma_i(t)$ be the identity of the server that queue $i$ chooses at time $t$. Note that if queue $i$ has no packets at time $t$, then $S^{i,j}_t=0$. Then the \textbf{regret} of queue $i$ on this window, denoted $\text{Reg}_i(w)$, is defined as 
\begin{equation}
    \text{Reg}_i(w)\triangleq \max_{j\in [m]} \sum_{t=1}^w S^{i,j}_t-\sum_{t=1}^{w} S^{i,\sigma_i(t)}_t.
\end{equation}
\end{defn}

That is, $\text{Reg}_i(w)$ of queue $i$ on some fixed window of length $w$ is defined to be the (random) difference between the number of packets queue $i$ cleared on these $w$ periods compared to the backward-looking number of packets she would have cleared had she simply always sent to the best single server, \emph{where the comparison is in hindsight to the best single server on the realized sample path, not to a counterfactural sample path where the queue always chose that server}. Note that all these random variables are with respect to the same sample path; the $S^{i,j}_t$ will depend on all previous randomizations and choices by the queues, as these implicitly yield the priorities of the queues. We make the following assumption on the regret of queuing strategies:

\begin{assumption}[Queues satisfy high-probability no-regret]
\label{assumption:learning}
All queues select servers using a strategy or algorithm satisfying the following \textbf{no-regret} guarantee: given \emph{fixed} $\delta\in (0,1)$ and a fixed window with length $w$, the regret $\text{Reg}_i(w)$ of queue $i$ on this given window of $w$ consecutive time steps satisfies
    $\text{Reg}_i(w)\leq \varphi_{\delta}(w)$
with probability at least $1-\delta$ only over their own randomizations during this window, where $\varphi_{\delta}(w)=o(w)$ is some explicit function. Here, $o(\cdot)$ hides constant factors depending on $\delta$ and $m$, but not $w$. 

Moreover, we require that the choices of the queue depend only on their past bandit feedback and their past history of ages, but not on their history of queue sizes.
\end{assumption}

For instance, this assumption holds with EXP3.P.1 with the form of the regret scaling like $\sqrt{w\ln(m w/\delta)}=o(w)$ \cite{auer2002nonstochastic}. Note that this high-probability guarantee is possible in our setting even in the priority model where the random variables of success at each server from the perspective of each queue at each time step depend on all previous actions (via timesteps and priorities), as well as the actions of the other queues in the current time period; see for instance the discussion in Section 9 of Auer, et al \cite{auer2002nonstochastic}. This property is standard and necessary in applying learning algorithms to multi-player games. Using the freezing technique of \cite{LykourisST_COLT18} for EXP3.P.1, one can ensure that such a guarantee holds simultaneously for each window of this length, and not only a fixed window, so the players would not have to be aware which window of size $w$ is relevant for our analysis. This is true as freezing guarantees that the probabilities associated with all arms remain high enough throughout the algorithm, which allows us to adapt the classical no-regret analysis starting at any time step for the window of the next $w$ time steps.

\headder{Dual Model: via Deferred Decision.}
To prove our main learning result, we will use the principle of  deferred-decisions to give an alternate description of the standard system given in Section \ref{section:standardsystem}. Suppose that in the standard model, each queue chooses which server to send to at time $t$ only depending on past feedback and their current oldest timestamp, but \emph{not on $Q^i_t$}. In this case, we  can  equivalently characterize the evolution of this system  keeping only the oldest timestamp of a packet at each queue. To do this, instead of randomly generating new packets at each time step according to a Bernoulli process, each queue only maintains the timestamp of their current oldest unprocessed packet. Once this packet is successfully cleared, the new current oldest unprocessed packet has timestamp generated by sampling a geometric random variable with parameter $\lambda_i$ and adding this to the timestamp of the just-completed packet. If this number exceeds to current time step $t$, this corresponds to having processed all packets that arrived before the current time step, and receiving the next packet in the future. 

We will call this random process the \textbf{dual process}. Because the gap between successes in repeated independent $\text{Bern}(\lambda_i)$ trials is given by a $\text{Geom}(\lambda_i)$ random variable, the standard and dual processes can be completely coupled, as described below. Concretely, when the queues use strategies with the above property, the dual process can be described using the following notation:
\begin{enumerate}
    \item Time progresses in discrete steps $t=0,1,\ldots$.
    \item At each time $t$, $T^{i'}_t$ is the timestamp of oldest unprocessed packet of queue $i$ at time $t$.
    \item $T^i_t=\max\{0,t-T^{i'}_t\}$ is the age of the current oldest packet of queue $i$ in relation to the current time step $t$. That is, $T^i_t$ measures how old the current oldest unprocessed packet for queue $i$ is. We say $T^i_t$ is the \textbf{age} of queue $i$ at time $t$.\footnote{Note that while $T^i_t\geq 0$ by definition, it is possible that $T_t^{i'}>t$. The interpretation of this is that the queue has cleared all of her packets at time $t$ and will receive her next one at time $t=T_t^{i'}$, or equivalently, in $T_t^{i'}-t$ steps in the future from the perspective at time $t$.}
    \item Queue $i$ can send a packet in this time step if $t-T^{i'}_t\geq 0$. If queue $i$ successfully clears her packet, set $T^{i'}_{t+1}= T^{i'}_t+X^i$, where $X^i\sim \text{Geom}(\lambda_i)$ is independent of all past events, and otherwise does not change.
\end{enumerate}

In general, we will write  $\mathbf{T}_t=(T^1_t,\ldots,T^n_t)\in \mathbb{N}^n$ for the vector of current ages of oldest packets. To see the equivalence, consider any standard queuing system with Bernoulli random variables $\{B^i_t\}_{i\in [n],t\geq 0}$ for packet generation. Then, to get a coupled dual system for the same system, use a sequence $\{G^i_j\}_{i\in [n],j\geq 0}$ with the interpretation that $G^i_j\sim \text{Geom}(\lambda_i)$ is the size of the $j$th gap between successes in the $B^i_t$. When queue $i$ clears her $j$th packet, her new oldest timestamp increases by $X^i=G^i_j$ as described above. As such gaps between timestamps in the standard model have $\text{Geom}(\lambda_i)$ distributions, the dual system gives the ages of each queue in the standard system at all times and gives an explicit coupling. 

The key feature is that, under the assumption that queues choose servers at time $t$ only based on at most the $T^i_t$, not on $Q^i_t$, all choices by queues are the same conditioned on just the current timestamp and past feedback as it is conditioned on all the past information in the standard model (which includes arrivals received after the current oldest packet). That is, if $\mathcal{F}_t$ denotes the information available to the queues in the standard model at time $t$, and $\mathcal{G}_t$ for the dual model, then all choices by the queues at time $t+1$ are the same conditioned on either history. The point of doing so is that $G^i_j$ will be independent of $\mathcal{G}_t$ until the queue clears her $j$th packet (namely, the timestamp of queue $i$'s $j+1$th packet is not known until the time queue $i$ clears her $j$th packet).

In the dual system, we define stability in the same way as before:
\begin{defn}
The dual system is said to be \textbf{strongly stable} if, for any fixed $r\geq 0$, $\mathbb{E}[(\sum_{i=1}^n T^i_t)^r]\leq C_r$ where $C_r$ is a fixed constant depending only on $r$, not on $t$.

The dual system is \textbf{almost surely stable} if, for any $c>0$, almost surely
\begin{equation*}
    \sum_{i=1}^n T^i_t = o(t^c).
\end{equation*}
\end{defn}

Because heuristically $Q_t^i\approx \lambda_i T_t^i$, it is intuitive that our notions of strong stability are equivalent whenever both systems correspond to the same random process. Indeed, this is the case. Moreover, strong stability implies almost sure stability. The basic idea is to use Markov's inequality and the Borel-Cantelli lemma along an appropriately chosen subsequence of times. We defer this equivalence and implication to the Appendix:
\begin{lem}
\label{lem:equivalence}
If the standard and dual models characterize the same queuing dynamics, then strong stability in the standard system is equivalent to strong stability in the dual system. 

Moreover, if this holds, then strong stability in either system implies almost sure stability. 
\end{lem}

As these are completely coupled processes, with same stability properties, it is natural to wonder what we have gained from focusing on the ages of queues rather than their overall sizes. We discuss this further in Remark \ref{remark:dual}.

\subsection*{Stability of No-Regret Queuing Systems}

Our main result is the following theorem which asserts that if all queues are no-regret with high-probability over sufficiently large windows, then the queuing system is strongly stable.

\begin{thm}
\label{thm:main}
Suppose that Assumption \ref{assumption:feasibility} holds for the dual queuing system with parameter $\eta$. Set the following parameters:
$\delta = \frac{\eta}{8},
    \epsilon=\frac{\delta\mu_1}{4n},$ and $
    \epsilon_i = \frac{\epsilon}{\lambda_i}$
for $1\leq i\leq n$.
Let $w$ be large enough so that the following holds\footnote{Note that this is possible as $\varphi_{\gamma}(w)=o(w)$ for any fixed $\gamma$, as well as the exponential decay of the bounds in (\ref{eq:geombound}) and (\ref{eq:bernbound}) in $w$.}:
\begin{equation}
\label{eq:wregassumption}
    n\varphi_{\frac{\eta}{128n}}(w)+n\leq \frac{w\delta\mu_1}{4},
\end{equation}
and over $w$ steps of our process the sum of the geometric variables of subsequent packet arrivals, and the sum of the Bernoulli server successes concentrate around their expectation with an error probability of at most $\eta/128$ with the above values of $\delta,\epsilon_i,\bm{\lambda}$, and $\bm{\mu}$. (See the required inequalities at (\ref{eq:geombound}) and (\ref{eq:bernbound}).)

Then, if each queue satisfies Assumption \ref{assumption:learning} on each consecutive time interval of length $w$ with probability at least $1-\frac{\eta}{128n}$, then the random process $\mathbf{T}_t$ under these dynamics is strongly stable.
\end{thm}

The technical tool we use to establish the stability of our system in Theorem \ref{thm:main} is the following result of Pemantle and Rosenthal:
\begin{thm}
[Theorem 1 in \cite{pemantle1999moment}]
\label{thm:pemantle}
Let $X_1,X_2,\ldots$ be a sequence of nonnegative random variables with the property that
\begin{enumerate}
    \item There exists constants $\alpha,\beta>0$ such that if $x_n>\beta$, then 
    \begin{equation*}
        \mathbb{E}[X_{n+1}-X_n\vert \mathcal{F}_n]<-\alpha,
    \end{equation*}
    where the $\sigma$-algebra $\mathcal{F}_n$ includes the history $\sigma(X_0,\ldots,X_n)$ 
    until period $n$ and $X_n=x_n$. 
    
    \item There exists $p>2$ and $\theta>0$ a constant such that for any history, 
    \begin{equation*}
        \mathbb{E}[\vert X_{n+1}-X_n\vert^p\vert \mathcal{F}_n]\leq \theta.
    \end{equation*}
\end{enumerate}
Then, for any $0<r<p-1$, there exists an absolute constant $M=M(\alpha,\beta,\theta,p,r)$ not depending on $n$ such that $\mathbb{E}[X_n^r]\leq M$ for all $n$.
\end{thm}

To apply this theorem, we must define an appropriate potential function of queue ages that satisfies the negative drift and bounded moments condition. We define for $\tau\in \mathbb{N}$ the following potential functions that will feature prominently in the proof:
\begin{gather}
\Phi_{\tau}(\mathbf{T}_t)\triangleq \sum_{i\in [n]:T^i_t\geq \tau} \lambda_i(T^i_t-\tau),\\
    \Phi(\mathbf{T}_t)\triangleq \sum_{\tau=1}^{\infty}\Phi_{\tau}(\mathbf{T}_t)=\sum_{\tau=1}^{\infty}\sum_{i\in [n]:T^i_t\geq \tau} \lambda_i(T^i_t-\tau)=\frac{1}{2}\sum_{i=1}^n \lambda_iT^i_t(T^i_t-1).
\end{gather}

This potential function will be useful for the analysis because it isolates the contribution of clearing old packets at each age simultaneously. We now turn to the proof of Theorem \ref{thm:main}.\\

\begin{proof}
To apply Theorem \ref{thm:pemantle}, we define the stochastic process $Z_0,Z_1,\ldots$ by
\begin{equation*}
    Z_{\ell} = \sqrt{\Phi(\mathbf{T}_{w\cdot \ell})}.
\end{equation*}
That is, $Z_{\ell}$ is the ``snapshot'' of the potential function $\sqrt{\Phi}$ when evaluated on $\mathbf{T}_{w\cdot \ell}$ that occurs every $w$ steps. The filtration is given by $\mathcal{H}_{\ell}=\mathcal{G}_{\ell\cdot w}$, where $\mathcal{G}_t$ is the corresponding information of the dual system at time $t$ available to the queues.

\headder{Summary of the Main Ideas:} 
Before we go though the detailed proof, we offer an outline of the main ideas. To establish the negative drift, we will focus on the $w$-long interval between two $Z_{\ell}$ and $Z_{\ell+1}$. In this $w$-long window, we use the no-regret condition, as well a concentration bounds on behavior of queues and servers. The main idea of the proof is to consider all \jgedit{queues} that have remained old throughout the period. A server either clears many such old packets, or many times during this period no old packet is sent to it. In the second case, we can use the no-regret condition for any queue that still has very old packets, as they would have priority at the server, so these bounds in tandem will imply that many old packets must have cleared. To aid the analysis, we also lower bound the total decrease in ages from clearing packets on this window \emph{before} accounting for the $w$ extra steps of aging, only accounting for this at the end; this allows us to consider the clearing process and aging from time passage separately. Finally, when concentration or the no-regret condition fails, we can trivially upper bound what this
contributes to the expected drift and this
will be subsumed by the low probability that this occurs in the overall expectation.

To establish the bounded moment condition, it is important to consider the dual process, as Theorem \ref{thm:pemantle} also requires that the change 
cannot be too large for any history. See Remark \ref{remark:dual} for more details.

\headder{Organizing Randomness.}
Let us first set up how we model the actual queuing process on each consecutive window of $w$ steps between $Z_{\ell}$ and $Z_{\ell+1}$ for the probabilistic analysis. In the spirit of ``organizing randomness," at step $\ell$ of this process (step $\ell\cdot w$ of the actual queuing process), sample up front an independent geometric ensemble $\{G_{i,j}\}_{i\in [n],j\in [w]}$ with
\begin{equation*}
    G_{i,j}\sim \text{Geom}(\lambda_i), j=1,\ldots,w
\end{equation*}
as well as an independent Bernoulli ensemble $\{S_{i,j}\}_{i\in [m],j\in [w]}$ with 
\begin{equation*}
    S_{i,j}\sim \text{Bern}(\mu_i), i=1,\ldots,w.
\end{equation*}
The interpretation is that the $S_{i,j}$ are random indicators if the $i$th server is able to clear a packet, \emph{regardless of whether a packet is sent there}, at the $j$th step of this block of $w$ steps. The $G_{i,j}$ have the interpretation that, when queue $i$ clears her $j$th packet on this window, her age decreases by $G_{i,j}$ (without accounting for the aging from passage of time). Crucially, as queue $i$ clears packets on this window of $w$ steps, her age decreases \emph{by a sum of a prefix of $G_{i,1},\ldots, G_{i,w}$} (before accounting for aging as time passes). Observe that this independence arises precisely because of the independence of the geometric ensemble of timestamp differences from the filtration $\mathcal{G}_{\ell\cdot w}$ of the dual system that only conditions on past feedback and the realized past sequence of oldest timestamps.

Now, observe that by our choices of parameters $\delta$ and $w$, we have that with probability at least $1-\eta/64$ that all of the following ``good'' events hold on this window:
\begin{equation}
\label{eq:goodgeom}
    \bigg\vert \sum_{j=1}^k G_{i,j}-\frac{k}{\lambda_i}\bigg\vert <\epsilon_i w\quad \forall i\in [n],k\in [w]
\end{equation}
\begin{equation}
\label{eq:goodbern}
    \sum_{i=1}^k\sum_{j=1}^w S_{i,j}>w(1-\delta)\sum_{i=1}^k \mu_i \quad \forall k\in [m]
    \end{equation}
    \begin{equation}
    \label{eq:goodreg}
    n+\sum_{i=1}^n\text{Reg}_i(w)\leq n+n\varphi_{\eta/128}n(w)\leq \frac{w\delta\mu_1}{4}
\end{equation}
This simply follows from our assumption that $w$ was chosen large enough so that the first two lines hold with probability at least $1-\eta/128$ via Corollary \ref{cor:geom3} and Lemma \ref{lem:bern} in the Appendix, the fact that the no-regret bound held with probability at least $1-\eta/128n$ for each queue, and taking a union bound. Notice that (\ref{eq:goodgeom}) asserts that \emph{every prefix of each of the geometric ensembles is additively not too far from the expectation, relative to $w$}.\\

\headder{Threshold Value for $Z_{\ell}$.} We will show that under the threshold assumption that 
\begin{equation}
\label{eqn:threshold}
    Z_{\ell}> \frac{w}{\sqrt{2\lambda_n}}\max\bigg(\frac{8}{\eta}\bigg(\sum_{i=1}^n \lambda_i\bigg),16n^2\bigg),
\end{equation}
then the drift condition holds. We will later use the following simple claim:
\begin{claim}
\label{claim:threshworks}
Under this assumption, both of the following statements hold:
\begin{enumerate}
    \item There exists some $i\in [n]$ such that $\lambda_i T^i_{\ell\cdot w}>16nw$.
    \item $\sum_{i=1}^n \lambda_i T^i_{\ell\cdot w}\geq \frac{8}{\eta}w\sum_{i=1}^n \lambda_i$.
\end{enumerate}
\end{claim}
\begin{proof}
(\ref{eqn:threshold}) immediately implies by definition of $Z_{\ell}$ and $\Phi$ that
\begin{equation*}
    \sqrt{\sum_{i=1}^n \lambda_i (T^i_{\ell\cdot w}})^2\geq \frac{w}{\sqrt{\lambda_n}}\max\bigg(\frac{8}{\eta}\bigg(\sum_{i=1}^n \lambda_i\bigg),16n^2\bigg).
\end{equation*}
From Lemma \ref{lem:norms}, this implies that
\begin{equation*}
    \sum_{i=1}^n \lambda_i T^i_{\ell\cdot w}\geq w\max\bigg(\frac{8}{\eta}\bigg(\sum_{i=1}^n \lambda_i\bigg),16n^2\bigg),
\end{equation*}
from which both parts follow, the first from averaging.
\end{proof}

\headder{Using no-regret to bound the number of old packets cleared.} Continuing with the proof, we first analyze what happens on the ``good'' event of (\ref{eq:goodgeom},\ref{eq:goodbern},\ref{eq:goodreg}). Let $\tau_i$ be the age of the oldest unprocessed packet of queue $i$ at the end of this window of $w$ steps, \emph{measured with respect to the beginning of the window without accounting for the $w$ steps of aging}. If queue $i$ cleared all her packets that were received before the beginning of this window, then we say $\tau_i=0$. Let $J_{\tau}$ be the set of queues that at the end of the $w$ steps still have packets that are at least $\tau$-old with respect to the beginning of the window. Let $X^{\tau}_{i,j}$ be an indicator variable that some packet that was at least $\tau$-old with respect to the beginning of the considered interval was sent to server $i$ at the $j$th step in this window. As such queues in $J_{\tau}$ evidently have packets that are at least $\tau$-old throughout this interval, priority and the regret bound (\ref{eq:goodreg}) implies that the number of packets cleared by any such queue is at least, for any server $i\in [m]$
\begin{equation}
\label{eqn:learning}
    \sum_{
    j=1}^w S_{i,j}(1-X^{\tau}_{i,j})-\varphi_{\eta/128n}(w).
\end{equation}
This is simply because a queue that is always at least $\tau$-old throughout the interval would succeed on any server $i$ that is successful on a time step (as indicated by $S_{i,j}$) where no $\tau$-old packets were sent there.

Let $N_{\tau}$ be the number of packets that were at least $\tau$-old with respect to the beginning of the interval that were cleared in the interval and $N^i_{\tau}$ the number of such packets cleared by queue $i$. Then we clearly have
\begin{equation}
\label{eqn:learning2}
    N_{\tau}=\sum_{i\in [m]}\sum_{j\in [w]} S_{i,j}X^{\tau}_{i,j}\geq \sum_{i=1}^{\min\{m,\vert J_{\tau}\vert\}} \sum_{j\in [w]} S_{i,j}X^{\tau}_{i,j}.
\end{equation}
As every packet processed by queues in $J_{\tau}$ contribute to $N_{\tau}$, by instantiating (\ref{eqn:learning}) for each queue in $J_{\tau}$ with each of the top $\min\{m,\vert J_{\tau}\vert\}$ servers and summing, we also obtain
\begin{equation}
    \min\{m,\vert J_{\tau}\vert\}\cdot N_{\tau}\geq \min\{m,\vert J_{\tau}\vert\}\sum_{i\in J_{\tau}} N^i_{\tau}\geq \vert J_{\tau}\vert \sum_{i=1}^{\min\{m,\vert J_{\tau}\vert\}}\bigg(\sum_{j=1}^w S_{i,j}(1-X_{i,j}^{\tau})-\varphi_{\eta/128n}(w)\bigg).
\end{equation}

Multiplying (\ref{eqn:learning2}) by $\vert J_{\tau}\vert$ and summing with the previous equation, we obtain
\begin{align}
    N_{\tau}&\geq\bigg(\frac{\vert J_{\tau}\vert}{\vert J_{\tau}\vert+\min\{m,\vert J_{\tau}\vert\}}\bigg)\sum_{i=1}^{\min\{m,\vert J_{\tau}\vert\}}\bigg(\sum_{j=1}^w S_{i,j}-\varphi_{\eta/128n}(w)\bigg)\\
    &\geq \frac{1}{2}\sum_{i=1}^{\min\{m,\vert J_{\tau}\vert\}}\sum_{j\in [w]} S_{i,j}-n\varphi_{\eta/128n}(w)\\
    \label{eq:packetlb}
    &\geq \frac{1}{2}w(1-\delta)\sum_{i=1}^{\min\{m,\vert J_{\tau}\vert\}}\mu_i-n\varphi_{\eta/128n}(w),
\end{align}
where the last inequality uses (\ref{eq:goodbern}).\\

\headder{Bounding the expected drift in $\Phi$ assuming the ``good event".}
Observe that from the construction of $\Phi_{\tau}$, when queue $i$ manages to process a packet that is at least $\tau$-old, $\Phi_{\tau}$ decreases either by $\lambda_i G_{i,j}$ for some $j$ if the new age remains above $\tau$, or the term vanishes in which case $\Phi_{\tau}$ may decrease by less. Crucially, this latter possibility can only happen at most once. Again, write $N^i_{\tau}$ for the number of packets that queue $i$ clears during this interval that are at least $\tau$-old. Then as $\sum_{i=1}^n N^i_{\tau}=N_{\tau},$
the decrease in $\Phi$ from $\Phi_{\tau}$, denoted $\Delta_{\tau}$, is
\begin{align*}
    \Delta_{\tau}&\geq \sum_{i=1}^n \lambda_i\bigg(\sum_{j=1}^{N_{\tau}^i-1} G_{i,j}\bigg)\\
    &\geq \sum_{i=1}^n \lambda_i\bigg(\frac{N_{\tau}^i-1}{\lambda_i}-\epsilon_i w\bigg) && \text{(by (\ref{eq:goodgeom}))}\\
    &=\sum_{i=1}^n (N_{\tau}^i-1-\lambda_i\epsilon_i w)\\
    &= N_{\tau}-n-n\epsilon w &&\text{(by definition of $N_{\tau}$ and $\epsilon_i$)}\\
    &\geq \frac{1}{2}w(1-\delta)\sum_{i=1}^{\min\{m,\vert J_{\tau}\vert\}}\mu_i-n\varphi_{\eta/128n}(w)-n-n\epsilon w &&\text{(by (\ref{eq:packetlb})})\\
    &\geq\frac{1}{2}w(1-\delta)\sum_{i=1}^{\min\{m,\vert J_{\tau}\vert\}}\mu_i-\frac{\delta w\mu_1}{2} &&\text{(by (\ref{eq:wregassumption}) and definition of $\epsilon$)}\\
    &\geq \frac{1}{2}w(1-2\delta)\sum_{i=1}^{\min\{m,\vert J_{\tau}\vert\}}\mu_i.
\end{align*}
Summing over all $\tau$, the decrease in $\Phi$ before considering aging is at least
\begin{align*}
    \Delta\Phi&\triangleq\sum_{\tau=1}^{\infty} \Delta_{\tau}\\
    &\geq \frac{1}{2}w(1-2\delta)\sum_{\tau=1}^{\infty}\sum_{i=1}^{\min\{m,\vert J_{\tau}\vert\}} \mu_i\\
    &=\frac{1}{2}w(1-2\delta)\sum_{i=1}^{\min\{m,n\}} \tau^{(i)}\mu_i
\end{align*}
where $\tau^{(i)}$ is the $i$'th largest of the $\tau_j$.\\

\headder{Effect of Aging.}
We now account for the increase due to aging by $w$ over the course of this interval. The increase in $\Phi$ from this is upper bounded by
\begin{equation*}
    \frac{1}{2}\sum_{i=1}^n \lambda_i(\tau_i+w)(\tau_i+w-1)-\frac{1}{2}\sum_{i=1}^n \lambda_i\tau_i(\tau_i-1)=w\sum_{i=1}^n \lambda_i \tau_i+\frac{1}{2}\sum_{i=1}^n \lambda_iw(w-1)\leq w\sum_{i=1}^n \lambda_i \tau_i+\frac{1}{2}\sum_{i=1}^n \lambda_iw^2.
\end{equation*}
Note that this is only exact for those $\tau_i$ that are nonzero, while is an upper bound for those that are zero. Combining these potential changes, we see that the potential decrease is at least
\begin{equation*}
    \frac{1}{2}w(1-2\delta)\sum_{i=1}^{\min\{m,n\}} \tau^{(i)}\mu_i-\bigg(w\sum_{i=1}^n \lambda_i \tau_i+\frac{1}{2}\sum_{i=1}^n \lambda_iw^2\bigg)\geq \frac{\eta}{2}w\sum_{i=1}^n \tau_i\lambda_i-\frac{1}{2}\sum_{i=1}^n \lambda_iw^2,
\end{equation*}
as $2\delta=\eta/4$ and using Assumption \ref{assumption:feasibility} with the fact that the product of two nonnegative sequences is maximal when both are in the same sorted order (see Lemma \ref{lem:dominancelem}).\\

\headder{Relating $\tau_i$ and $T^i_{\ell\cdot w}$.}
We now need the following claim that relates the $\tau_i$ and $T^i_{\ell\cdot w}$:
\begin{claim}
\label{claim:claim2}
If there exists an $i\in [n]$ such that $\lambda_i T^i_{\ell\cdot w}>16nw$ under the good event assumptions, then 
\begin{equation*}
    \sum_{i=1}^n \lambda_i\tau_i\geq \frac{1}{2}\sum_{i=1}^n \lambda_i T^i_{\ell\cdot w}.
\end{equation*}
\end{claim}
\begin{proof}
By (\ref{eq:goodgeom}), we must have
\begin{equation*}
    \lambda_i\tau_i\geq \lambda_i T^i_{\ell\cdot w}-w-\lambda_i\epsilon_i w,
\end{equation*}
even if queue $i$ clears a packet every step in the window.
Observe that if $\lambda_i T^i_{\ell\cdot w}\geq 8w$, then 
\begin{equation*}
    w+\lambda_i\epsilon_i w\leq 2w\leq \frac{1}{4}\lambda_i T^i_{\ell\cdot w},
\end{equation*}
and so 
\begin{equation*}
    \lambda_i\tau_i\geq \frac{3}{4}\lambda_i T^i_{\ell\cdot w}.
\end{equation*}
We also have
\begin{equation*}
    \frac{1}{2}\sum_{i:\lambda_iT^i_{\ell\cdot w}< 8w} \lambda_iT^i_{\ell\cdot w}<4nw.
\end{equation*}
In particular, if there exists some $i$ such that $\lambda_iT^i_{\ell\cdot w}>16nw$, then 
\begin{equation*}
    \lambda_i\tau_i\geq \frac{3}{4}\lambda_iT^i_{\ell\cdot w}> \frac{1}{2}\lambda_iT^i_{\ell\cdot w}+4nw\geq \frac{1}{2}\lambda_iT^i_{\ell\cdot w}+\frac{1}{2}\sum_{i:\lambda_it_i< 8w} \lambda_iT^i_{\ell\cdot w}.
\end{equation*}
It follows that if this holds, then
\begin{equation*}
    \sum_{i=1}^n \lambda_i\tau_i=\sum_{i:\lambda_iT^i_{\ell\cdot w}\geq 8w} \lambda_i\tau_i+\sum_{i:\lambda_iT^i_{\ell\cdot w}< 8w} \lambda_i\tau_i\geq \frac{1}{2}\sum_{i=1}^n \lambda_iT^i_{\ell\cdot w},
\end{equation*}
as claimed.
\end{proof}

By Claim \ref{claim:threshworks}, the precondition of Claim \ref{claim:claim2} holds for our threshold value, so the decrease in $\Phi$ on this good event is at least
\begin{equation*}
    \frac{\eta}{4}w\sum_{i=1}^n \lambda_iT^i_{\ell\cdot w}-\frac{1}{2}\sum_{i=1}^n \lambda_iw^2\geq \frac{\eta}{8}w\sum_{i=1}^n \lambda_iT^i_{\ell\cdot w},
\end{equation*}
where the inequality is also Claim \ref{claim:threshworks}. Translating this into the  decrease in $\sqrt{\Phi}$, Fact \ref{fact:root2} implies that the contribution towards the expected decrease on this event, which occurs with probability at least $1-\eta/64\geq 1/2$, is at least
\begin{equation}
\label{eqn:decrease}
    \bigg(\frac{1}{2}\bigg)\frac{\frac{\eta}{8}w\sum_{i=1}^n \lambda_iT^i_{\ell\cdot w}}{2\cdot \sqrt{\frac{1}{2}\sum_{i=1}^n \lambda_iT^i_{\ell\cdot w}(T^i_{\ell\cdot w}-1)}}= \frac{\eta w\sum_{i=1}^n \lambda_i T^i_{\ell\cdot w}}{32\cdot\sqrt{\frac{1}{2}\sum_{i=1}^n \lambda_iT^i_{\ell\cdot w}(T^i_{\ell\cdot w}-1)}}
\end{equation}\\

\headder{Considering ``bad" events.} 
We now analyze the bad event where any of these assumptions fails: the worst case is that all queues clear no packets, and so each $T^i_{\ell\cdot w}$ increases by $w$ on the next $w$ steps. The increase in $\Phi$ is thus at most
\begin{equation*}
    \frac{1}{2}\sum_{i=1}^n \lambda_i(T^i_{\ell\cdot w}+w)(T^i_{\ell\cdot w}+w-1)-\frac{1}{2}\sum_{i=1}^n \lambda_iT^i_{\ell\cdot w}(T^i_{\ell\cdot w}-1)\leq w\sum_{i=1}^n \lambda_iT^i_{\ell\cdot w}+\frac{1}{2}\sum_{i=1}^n \lambda_iw^2\leq 2w\sum_{i=1}^n \lambda_i T^i_{\ell\cdot w},
\end{equation*}
where the last inequality is again Claim \ref{claim:threshworks}.
Translating to squareroots again, on this bad event which occurs with probability at most $\eta/64$, the contribution of increase to the expected change in $\sqrt{\Phi}$ is at most
\begin{equation}
\label{eqn:increase}
    \bigg(\frac{\eta}{64}\bigg)\frac{2w\sum_{i=1}^n \lambda_i T^i_{\ell\cdot w}}{2\cdot \sqrt{\frac{1}{2}\sum_{i=1}^n \lambda_iT^i_{\ell\cdot w}(T^i_{\ell\cdot w}-1)}}=\frac{\eta w\sum_{i=1}^n \lambda_i T^i_{\ell\cdot w}}{64\cdot \sqrt{\frac{1}{2}\sum_{i=1}^n \lambda_iT^i_{\ell\cdot w}(T^i_{\ell\cdot w}-1)}},
\end{equation}
by Fact \ref{fact:root1}. Summing (\ref{eqn:decrease}) and (\ref{eqn:increase}), it follows $\sqrt{\Phi}$ decreases in expectation by at least
\begin{equation*}
    \frac{\eta w\sum_{i=1}^n \lambda_i T^i_{\ell\cdot w}}{64\cdot \sqrt{\frac{1}{2}\sum_{i=1}^n \lambda_i T^i_{\ell\cdot w}(T^i_{\ell\cdot w}-1)}}\geq \frac{\eta w\sqrt{\lambda_n}}{64},
\end{equation*}
where the last inequality is Lemma \ref{lem:norms}. This proves that the drift condition holds for this stochastic process with the threshold given above.\\

\textbf{Bounded $p$th Moments:} The last thing to check to apply Theorem \ref{thm:pemantle} is show that the increments $Z_{\ell+1}-Z_{\ell}$ have conditionally bounded $p$th moments for each even integer $p\geq 2$ to obtain boundedness of our sequence in $L^r$ for all $r\geq 0$. But this is relatively straightforward: by the Triangle Inequality, it is easy to see that as random variables, the change in $T^i_{\ell\cdot w}$ is at most
\begin{equation*}
   G_i := \sum_{j=1}^w G_{i,j}.
\end{equation*}
Then the change in $\Phi$ is again at most
\begin{equation*}
    \frac{1}{2}\sum_{i=1}^n \lambda_i(T^i_{\ell\cdot w}+G_i)(T^i_{\ell\cdot w}+G_i-1)-\frac{1}{2}\sum_{i=1}^n \lambda_i(T^i_{\ell\cdot w})(T^i_{\ell\cdot w}-1)\leq \sum_{i=1}^n \lambda_i G_iT^i_{\ell\cdot w}+\frac{1}{2}\sum_{i=1}^n \lambda_iG_i^2,
\end{equation*}
as random variables. We treat two different cases separately:
\begin{enumerate}
    \item Suppose there does not exist $i\in [n]$ such that $\lambda_iT^i_{\ell\cdot w}>1$. Then the change in $\Phi$ is at most
    \begin{equation*}
        \sum_{i=1}^n G_i+\frac{1}{2}\sum_{i=1}^n \lambda_i G_i^2.
    \end{equation*}
    From Fact \ref{fact:root1}, this means the change in $\sqrt{\Phi}$ is upper bounded as random variables by 
    \begin{equation*}
        \sqrt{\sum_{i=1}^n G_i+\frac{1}{2}\sum_{i=1}^n \lambda_i G_i^2}.
    \end{equation*}
    Raising this to the $p=2k$ power, expanding, and taking expectations, this term is at most $C_{p,n,w}/\lambda_n^{2p}$ for some constant $C_{p,n,w}$ depending only on $n,w,$ and $p$ by Lemma \ref{lem:geomom}.
    
    \item Suppose there does exist $i\in n$ such that $\lambda_i T^i_{\ell\cdot w}>1$. We claim this implies that for all $j\in [n]$
    \begin{equation*}
        \frac{\lambda_j T^j_{\ell\cdot w}}{2\sqrt{\frac{1}{2}\sum_{i=1}^n\lambda_i T^i_{\ell\cdot w}(T^i_{\ell\cdot w}-1)}}\leq \sqrt{\lambda_j}
    \end{equation*}
    First, note that for any $i\in [n]$, $T^i_{\ell\cdot w}\geq 2$ implies
    \begin{equation}
    \label{eqn:t}
        \frac{1}{2}\lambda_i T^i_{\ell\cdot w}(T^i_{\ell\cdot w}-1)\geq \frac{1}{4} \lambda_i (T^i_{\ell\cdot w})^2,
    \end{equation}
    as can be confirmed from basic algebra. As $\lambda_i\leq 1/2$ by feasibility (as $\mu_1\leq 1$), our assumption implies $T^i_{\ell\cdot w}>2$, and so
    \begin{equation*}
        2\sqrt{\frac{1}{2}\sum_{i=1}^n \lambda_i T^i_{\ell\cdot w}(T^i_{\ell\cdot w}-1)}>1.
    \end{equation*}
    To prove the claim, we split into more cases: if $T^j_{\ell\cdot w}\leq 1/\sqrt{\lambda_j}$, the claim holds using the last inequality in the denominator. Otherwise, we must have $T^j_{\ell\cdot w}\geq 2$, in which case by (\ref{eqn:t}), 
    \begin{equation*}
        \frac{\lambda_j T^j_{\ell\cdot w}}{2\cdot \sqrt{\frac{1}{2}\sum_{i=1}^n\lambda_i T^i_{\ell\cdot w}(T^i_{\ell\cdot w}-1)}}\leq \frac{\lambda_j T^j_{\ell\cdot w}}{\sqrt{\lambda_j (T^j_{\ell\cdot w})^2}}= \sqrt{\lambda_j}.
    \end{equation*}
    Thus, in this case, we have
    \begin{equation*}
        \frac{\sum_{i=1}^n \lambda_i G_iT^i_{\ell\cdot w}+\frac{1}{2}\sum_{i=1}^n \lambda_iG_i^2}{2\cdot \sqrt{\frac{1}{2}\sum_{i=1}^n\lambda_i T^i_{\ell\cdot w}(T^i_{\ell\cdot w}-1)}}\leq \sum_{i=1}^n \sqrt{\lambda_i}G_i+\frac{1}{2}\sum_{i=1}^n \lambda_i G_i^2.
    \end{equation*}
    By Fact \ref{fact:root1}, this is an upper bound as random variables of the change in $\sqrt{\Phi}$, so taking $p=2k$ powers, expanding, and taking expectations, we get an upper bound of $C_{p,n,w}/\lambda_n^{2p}$ by Lemma \ref{lem:geomom} for some constant $C_{p,n,w}$ depending only on $n,w,p$.
\end{enumerate}

Therefore, Theorem \ref{thm:pemantle} applies to the random process $Z_{\ell}$, and we conclude that for each $r\geq 0$, there exists some absolute constant $C_r$ such that for all $\ell=0,1,\ldots,$
\begin{equation*}
    \mathbb{E}[Z_{\ell}^r]\leq C_r.
\end{equation*}
In particular, this means that for each $\ell\geq 0$,
\begin{equation*}
    \mathbb{E}\bigg[\bigg(\sqrt{\sum_{i=1}^n \lambda_iT^i_{\ell\cdot w}(T^i_{\ell\cdot w}-1)}\bigg)^r\bigg]\leq C.
\end{equation*}
To extend this to all $t\geq 0$ not necessarily of this form, it is clear that deterministically,
\begin{equation*}
\sqrt{\Phi(\mathbf{T}_t)}\asymp \|\mathbf{T}_t\|_{\lambda,2}\asymp \|\mathbf{T}_t\|_{\lambda,1}
\end{equation*}
up to additive and multiplicative constants as in Lemma \ref{lem:norms}, from which we can conclude
\begin{equation*}
    \mathbb{E}\bigg[\bigg(\sum_{i=1}^n \lambda_i T^i_t\bigg)^r\bigg]\leq C'_r
\end{equation*}
for some other constant $C'_r$ for each $t=\ell w$. Note that for
$\ell w\leq t<(\ell+1) w$, each term can increase by at most $w$ compared to the value at $\ell w$, and therefore we can conclude that for \emph{all} $t\geq 0$,
\begin{equation*}
    \mathbb{E}\bigg[\bigg(\sum_{i=1}^n T^i_t\bigg)^r\bigg]\leq C''_r
\end{equation*}
for some constant $C''_r$ independent of $t$. This concludes the proof of strong stability.

\end{proof}

\begin{remark}
\label{remark:dual}
This analysis crucially relies on using the dual system as opposed to the standard system. The reason is that the preconditions in Theorem \ref{thm:pemantle} must hold conditioned on any history, however implausible. In the standard system, this would require us to condition on too much. For instance, it is technically possible for there to be a queue with a very old packet, and yet have received no other packets until the current time step. While unlikely to actually ever happen, this is a perfectly valid potential history. In this case, clearing this packet would lead to arbitrarily large $p$th moment change, as her age would drastically decrease, and therefore the moment condition of Theorem \ref{thm:pemantle} would be violated. While intuitively this should only help the stability of the random process, the conditions in Theorem \ref{thm:pemantle} are surprisingly subtle, see the discussion in \cite{pemantle1999moment}. 

In general, even if that obstruction can be managed suitably, the extra conditional information in the standard system highly complicates the analysis, as then one must reason about the priorities of the packets that have already been received before the present time step, which could in principle be quite arbitrary. We avoid these complications in the dual system as it allows us to only condition on current ages and sample new packets as they come throughout the next $w$ steps, from which we can use concentration to argue that this process is well-behaved enough for our analysis to succeed. 
\end{remark}

We now provide a simple construction showing that a partial converse holds: $\frac{1}{2}$ is the best constant that can appear in Assumption \ref{assumption:feasibility} for a similar no-regret condition to be sufficient for stability as in Theorem \ref{thm:main}.
\begin{thm}
\label{thm:2tight}
Partition time $t=0,1,\ldots$ into consecutive windows, where the $k$th window has length $w_k=k^2$. Then there exists a family of queuing systems with $n$ queues and servers for each $n\geq 1$ satisfying Assumption \ref{assumption:feasibility} with $\frac{1}{2}+o_n(1)$ in place of $\frac{1}{2}$ with the following properties: almost surely, each queue has zero regret on all but at most finitely many of the windows, but the system is not strongly stable.
\end{thm}

The formal details are slightly technical, but the high-level idea is quite natural: for each $n\geq 1$, consider the following system on $n$ queues and $n$ servers where we set $\bm{\lambda}=(\frac{n+1}{n^2},\ldots,\frac{n+1}{n^2})$ and $\bm{\mu}=(1,\frac{n-1}{n^2},\ldots,\frac{n-1}{n^2})$. Consider the strategy where every queue always sends to the rate $1$ server. It is easy to see purely from expectations that the queue lengths are unbounded in expectation, as the sum of arrival rates strictly exceeds $1$. On the other hand, it is intuitive that this strategy will ``usually'' be zero-regret; if all the queues are similarly aged at the start of some window, then they should expect to clear roughly $1/n$ fraction of the time on this window using this strategy, which strictly exceeds what they would get at any other server. We use standard concentration arguments and the Borel-Cantelli lemma to argue that this situation will happen all but finitely many times almost surely, thereby obtaining the claim. We now give the formal argument:\\

\begin{proof}
Define $W_k=\sum_{i=1}^{k-1} w_i$. Note $W_k=\Theta(k^3)=\Theta(w_k^{3/2})$. Note that we are slightly abusing our convention here, as $W_k$ is not a random variable, rather just convenient notation. $W_k$ is the actual time step at the end of $k-1$ of the consecutive windows of length $w_i$ for $i=1,\ldots,k-1$. Note also that $W_{k+1}-W_k=w_k$.

For each $n\geq 1$, consider the following system on $n$ queues and $n$ servers: set $\bm{\lambda}=(\frac{n+1}{n^2},\ldots,\frac{n+1}{n^2})$ and $\bm{\mu}=(1,\frac{n-1}{n^2},\ldots,\frac{n-1}{n^2})$. Note that this system satisfies Assumption \ref{assumption:feasibility} with factor $\frac{1}{2}-o_n(1)$. We will consider the simple strategy where every queue always sends to the rate $1$ server. Note that under these dynamics, in expectation the total number of packets grows by $\frac{1}{n}$ with every step, and therefore this system is not strongly stable. What we must show is that almost surely, this fixed strategy is zero regret for every queue for all but finitely many of the windows.

As this strategy is oblivious, we may study the standard system. First we show almost sure concentration of the arrivals of new packets. Let $\{B^i_t\}_{i\in [n],t\geq 1}$ be the independent random variables for arrivals as usual. Now, for each queue $i\in [n]$ and $\ell\geq 0$, we have
\begin{equation}
    \Pr\bigg(\bigg\vert \sum_{t=1}^{\ell} B^i_t-\lambda_i\ell\bigg\vert\geq \sqrt{\ell \ln(\ell)}\bigg)\leq \frac{2}{\ell^2},
\end{equation}
where we use the additive form of the Chernoff bound. As the same holds for all queues, the probability this event happens for any of the $n$ queues is at most $2n/\ell^2$. As this is summable in $\ell$, we may sum over all $\ell\geq 1$ to deduce from the Borel-Cantelli lemma that almost surely, for all sufficiently large $\ell$, all $i\in [n]$ satisfy
\begin{equation}
\label{eq:queuebc}
   \sum_{t=1}^{\ell} B^i_t=\lambda_i\ell\pm O(\sqrt{\ell \ln(\ell)}).
\end{equation}
 Note that this also implies that almost surely, for all large enough $\ell$, $\sum_{t=1}^\ell\sum_{i=1}^n B^i_t\geq (1+\frac{1}{2n})\cdot \ell$ by the choice of $\lambda_i$. Observe also that under this fixed strategy where everyone always sends to the rate $1$ server, at most $\ell$ packets can be cleared by time $\ell$. 

Next, we show that almost surely, there is a large backup proportional to the current time period. Let $t_k$ be the last timestamp the rate $1$ server clears up to time $W_k$. As all queues send there under this fixed strategy, at this point, all queues only have packets that were received after $t_k$ by priority. On the one hand, it is not difficult to see that deterministically $t_k\geq W_k/n$ (equality happens in the worst case where every queue received a packet in every step up to $W_k$). On the other hand, in light of our results above, almost surely, for all but finitely many of the $k$,
\begin{equation}
\label{eq:timebc}
t_k<\frac{1}{1+\frac{1}{2n}} W_k=(1-\Omega(1))W_k.
\end{equation}
This is because at least $W_k$ packets have been received up to time $\frac{1}{1+\frac{1}{2n}} W_k$, and because the server can only have cleared at most $W_k$ packets up to time $W_k$, the oldest timestamp the server could have cleared by time $W_k$ can be at most this quantity.

Next, we show almost sure concentration of the nontrivial server success rates. Let $S^j_t$ be the indicator that server $j$ \emph{would} succeed at clearing a packet at time $t$ (regardless of if one is sent there; indeed, under this strategy no queue ever sends to $j\neq 1$). A similar application of the Chernoff bound and union bound with the Borel-Cantelli lemma implies that almost surely, for all but finitely many of the $k$, and for each server $j\in [n]$, we have
\begin{equation}
\label{eq:serverbc}
    \sum_{t=W_k+1}^{W_{k+1}} S^j_t=w_{k}\mu_i \pm O(\sqrt{w_k\ln w_k}).
\end{equation}
Note that the increasing nature of the $w_k$ is needed here for this to be valid (and in fact, this statement will be false with probability one if interval sizes are kept fixed by independence and the second Borel-Cantelli lemma).

Thus, almost surely, for all large enough $\ell$ and $k$, all of these events we have described happen simultaneously. As we know $t_k\geq W_k/n$, almost surely for large enough $k$, $t_k$ eventually exceeds the random time $\ell$ at which (\ref{eq:queuebc}) holds. Consider any subsequent window of length $w_k$. Our goal is to combine the above facts and show that on these windows, all queues get zero regret.

First, we show that each queue clears $(\frac{1}{n}-o(1))w_k$ packets on each such window. Let $c=\frac{1}{n\lambda_i}<1$ (note this is independent of $i$). We know by virtue of (\ref{eq:timebc}) that $t_k+w_k<(1-\Omega(1))W_k+w_k<W_{k}$; moreover, by virtue of (\ref{eq:queuebc}), and the fact $t_k\geq \ell$, we have that 
\begin{align*}
    \sum_{t=t_k+1}^{t_k+c\cdot w_k} B^i_t&=\sum_{t=1}^{t_k+c\cdot w_k} B^i_t-\sum_{t=1}^{t_k} B^i_t \\
    &= \frac{1}{n} \cdot w_k \pm O(\sqrt{(t_k+c\cdot w_k)\ln{(t_k+c\cdot w_k)}})\\
    &=\frac{1}{n} \cdot w_k \pm O(\sqrt{W_k\ln{W_k}})\\
    &=\frac{1}{n} \cdot w_k \pm \tilde{O}(w_k^{3/4}),
\end{align*}
where the last line uses the relationship between $W_k$ and $w_k$. As $t_k+w_k<W_k$, all of these packets were evidently received before the start of the given window, and therefore, every queue is backed up throughout the period, and by virtue of the previous equation, each queue has $\frac{1}{n}-o(1)$ fraction of the next $w_k$ packets that will be cleared by this top server on this window. Therefore, each queue clears at least $(\frac{1}{n}-o(1))\cdot w_k$ packets on such windows under this fixed strategy.

Finally, had any queue deviated on such a window to a single fixed low rate server, in light of (\ref{eq:serverbc}), she would have cleared
\begin{equation}
    \bigg(\frac{n-1}{n^2}+o(1)\bigg)\cdot w_k
\end{equation}
packets, which is linearly smaller than the amount she actually cleared. Therefore, almost surely, on all but finitely many of the windows, every queue actually has zero regret.

\end{proof}

\section{Discussion and Open Problems}
In this work, we have shown that high-probability versions of the no-regret property can lead to stability in appropriate queuing systems if there is only a constant factor of slack; however, the model specifications really are crucial, as evidenced by the alternate model.

There are many open questions: in the context of just this work, is the high-probability requirement in Assumption \ref{assumption:learning} necessary? Many no-regret algorithms only satisfy a no-regret bound in expectation, which is \emph{a priori} a weaker requirement; we do not see an immediate way to derive our results using just a regret bound that holds in expectation, as it is difficult to argue about various correlations that may arise. Furthermore, while the proof of Theorem \ref{thm:main} only really relies on the feature of having good enough regret on long enough intervals with high enough probability, to attain these precise parameters for all queues via a natural learning algorithm requires some mild synchronization between queues. For instance, algorithms like EXP3.P require the desired probability bound  and the length of the interval $w$ that they aim to not have regret on as an input to the learning algorithm. 
Is there a way to establish a similar result with more oblivious settings of natural learning algorithms, perhaps using a different analysis of this process? 

One natural direction is to extend the results here to more general queuing networks. For instance, a queue may need to choose a full path in a network, instead of directly send to a server. In these settings, richer feedback structures and action spaces are possible, as queues may receive feedback with a certain delay, or may have different available paths to route packets. Alternately, a packet may need to go through multiple queues before reaching the server or its destination, where each queue is running its own learning algorithm for forward packets.
Moreover, in Krishnasamy, et al \cite{krishnasamy2016learning}, the authors motivate a more relevant measure of \emph{queue-regret} that perhaps better describes the performance of learning in such systems. Is it possible to combine the analysis of this paper in the strategic setting with their more refined learning results for the learning problem in queues?

Beyond this setting, many natural strategic interactions hold this sort of ``carryover'' effect, where the results of previous interactions have a strong effect on the fundamental nature of the current interaction. We hope that some of the techniques and results here in a decently simple queuing model may serve as a preliminary step towards the study of such highly dependent interactions in more complicated settings. These sorts of infinitely repeated games also hold the potential for establishing qualitatively different forms of price of anarchy results. Here, the natural metric was a binary form of stability, which can only be formulated as a long-run phenomenon. Understanding the interplay between games and learning with these types of qualitative different objectives seems like a fruitful avenue for future work in this area.

\section{Acknowledgements}
We would like to thank Christos Papadimitriou for many great discussions, his useful insights and encouragement in pursuing this model in the early stages of this work.

\bibliography{LearningQueues}
\bibliographystyle{ieeetr}

\newpage 
\section{Appendix}
Here we collect various technical results and proofs that are used in the main part of the paper.

\subsection{Basic Inequalities}
\begin{fact}
\label{fact:root1}
Suppose $a,b,c\geq 0$ and further
\begin{equation*}
    a-b\leq c.
\end{equation*}
Then 
\begin{equation*}
    \sqrt{a}-\sqrt{b}\leq \min\bigg\{\frac{c}{2\sqrt{b}},\sqrt{c}\bigg\}.
\end{equation*}
\end{fact}
\begin{proof}
The first inequality arises from rearranging and concavity of the squareroot function:
\begin{equation*}
    \sqrt{a}\leq \sqrt{b}\sqrt{1+c/b}\leq \sqrt{b}(1+c/2b).
\end{equation*}
The second follows from assuming without loss of generality that $a\geq b$ and observing the claim is implied by
\begin{equation*}
    \sqrt{a}-\sqrt{b}\leq \sqrt{a-b},
\end{equation*}
which holds by squaring and simple algebra.
\end{proof}

\begin{fact}
\label{fact:root2}
Suppose $a,b,c\geq 0$. Then $a-b\geq c$ implies 
\begin{equation*}
\sqrt{a}-\sqrt{b}\geq \frac{c}{2\sqrt{a}}.
\end{equation*}
\end{fact}
\begin{proof}
\begin{equation}
    a-b=(\sqrt{a}-\sqrt{b})(\sqrt{a}+\sqrt{b})\geq c\implies \sqrt{a}-\sqrt{b}\geq \frac{c}{\sqrt{a}+\sqrt{b}}\geq \frac{c}{2\sqrt{a}}.
\end{equation}
\end{proof}

Recall that we defined the following two weighted $\ell_p$ norms on $\mathbb{R}^n$: $
    \|\mathbf{x}\|_{\bm{\lambda},1}\triangleq \sum_{i=1}^n \lambda_i \vert x_i\vert$ and $
    \|\mathbf{x}\|_{\bm{\lambda},2}\triangleq \sqrt{\sum_{i=1}^n \lambda_i x_i^2}.$ We will need the following simple relationship between the norms defined above that quantifies their equivalence:
\begin{lem}
\label{lem:norms}
For all $x\in \mathbb{R}^n$,
\begin{equation*}
    \sqrt{\lambda_n}\|x\|_{\bm{\lambda},2}\leq \|x\|_{\bm{\lambda},1}\leq \sqrt{\sum_{i=1}^n \lambda_i}\|x\|_{\bm{\lambda},2}.
\end{equation*}
\end{lem}
\begin{proof}
For the first inequality, 
\begin{align*}
    \|x\|_{\bm{\lambda},1}^2&=\sum_{i,j=1}^n \lambda_i\lambda_j\vert x_i\vert \vert x_j\vert\\
    &\geq \sum_{i=1}^n \lambda_i^2x_i^2\\
    &\geq \lambda_n \sum_{i=1}^n \lambda_ix_i^2\\
    &=\lambda_n\|x\|_{\bm{\lambda},2}^2.
\end{align*}

The second is a routine application of Cauchy-Schwarz:
\begin{equation*}
    \sum_{i=1}^n \lambda_i \vert x_i\vert=\sum_{i=1}^n \sqrt{\lambda_i}(\sqrt{\lambda_i}\vert x_i\vert)\leq \sqrt{\sum_{i=1}^n\lambda_i}\sqrt{\sum_{i=1}^n \lambda_i x_i^2}=\sqrt{\sum_{i=1}^n\lambda_i}\|x\|_{\bm{\lambda},2}.
\end{equation*}
\end{proof}

\subsection{Probability Tools}
We will use the following concentration results throughout the paper.

\begin{lem}[First Borel-Cantelli Lemma, Theorem 2.3.1 of \cite{durrett2019probability}]
Let $A_1,A_2,\ldots$ be a sequence of events with $\sum_{i=1}^{\infty}\Pr(A_i)<\infty$. Then with probability one at most finitely many of the $A_i$ occur.
\end{lem}

\begin{lem}[Azuma-Hoeffding]
\label{lem:azuma}
Let $\{\mathcal{F}_k\}_{k\leq n}$ be any filtration and let $A_k,B_k,\Delta_k$  satisfy the following conditions:
\begin{enumerate}
    \item $\Delta_k$ is $\mathcal{F}_k$-measurable and $\mathbb{E}[\Delta_k\vert \mathcal{F}_{k-1}]=0$. That is, the $\Delta_k$ form a martingale difference sequence.
    \item $A_k,B_k$ are $\mathcal{F}_{k-1}$-measurable and satisfy $A_k\leq \Delta_k\leq B_k$ almost surely.
\end{enumerate}
Then 
\begin{equation}
    \Pr\bigg(\sum_{k=1}^n \Delta_k \geq t\bigg)\leq \exp\bigg(\frac{-2t^2}{\sum_{k=1}^n \|B_k-A_k\|_{\infty}}\bigg).
\end{equation}
\end{lem}
\begin{lem}[Etemadi, Theorem 22.5 in \cite{billingsley2008probability}]
\label{lem:etemadi}
Suppose $X_1,\ldots,X_n$ are independent random variables. Then for any $x\geq 0$,
\begin{equation*}
    \Pr\bigg(\max_{1\leq k\leq n} \vert \jgedit{Z}_k\vert\geq 3x\bigg)\leq 3\max_{1\leq i\leq n}\Pr\big(\vert \jgedit{Z}_k\vert\geq x\big),
\end{equation*}
where \jgedit{$Z_k$} is the $k$'th partial sum of the $X_i$, i.e. \jgedit{$Z_k$}$=\sum_{i=1}^k X_i$.
\end{lem}
\begin{lem}[Theorem 1 in \cite{witt2013fitness}]
\label{lem:geom1}
Let $X_1,\ldots,X_n$ be i.i.d. $\text{Geom}(\lambda)$ random variables, so that $\mathbb{E}[X_i]=\frac{1}{\lambda}$. Let $s=\frac{n}{\lambda^2}$. Then for all $\delta>0$,

\begin{equation*}
    \Pr\bigg( \jgedit{Z_n}-\frac{n}{\lambda}<-\delta\bigg)\leq \exp\bigg(\frac{-\delta^2}{2s}\bigg),
\end{equation*}
and
\begin{equation*}
    \Pr\bigg( \jgedit{Z_n}-\frac{n}{\lambda}>\delta\bigg)\leq \exp\bigg(\frac{-\delta}{4}\min\{\delta/s,\lambda\}\bigg)
\end{equation*}
where $\jgedit{Z_n}=\sum_{i=1}^n X_i$.
\end{lem}

\begin{corollary}
\label{cor:geom2}
Under the assumptions and notation of Lemma \ref{lem:geom1}, for any $\epsilon\in [0,1]$,
\begin{equation*}
    \Pr\bigg(\max_{1\leq j\leq n} \bigg\vert \jgedit{Z_j}-\frac{j}{\lambda}\bigg\vert > \frac{\epsilon n}{\lambda}\bigg)\leq 6 \exp\bigg(\frac{-\epsilon^2n}{36}\bigg).
\end{equation*}
\end{corollary}
\begin{proof}
First apply Lemma \ref{lem:geom1} for each partial sum $\jgedit{Z_j}$ and $\delta=\epsilon n/\lambda$. By considering the cases $j\leq \epsilon n$ and $j>\epsilon n$ respectively, it follows for all $j\leq n$,
\begin{equation*}
    \min\{\delta/s,\lambda\}\geq \epsilon \lambda.
\end{equation*}
Lemma \ref{lem:geom1} now implies for all $j\leq n$,
\begin{gather*}
    \Pr\bigg( \jgedit{Z_j}-\frac{j}{\lambda}<-\frac{\epsilon n}{\lambda}\bigg)\leq \exp\bigg(\frac{-\epsilon^2n^2}{2j}\bigg)\leq \exp\bigg(\frac{-\epsilon^2n}{4}\bigg)
\end{gather*}
and similarly
\begin{gather*}
    \Pr\bigg( \jgedit{Z_j}-\frac{j}{\lambda}>\frac{\epsilon n}{\lambda}\bigg)\leq \exp\bigg(\frac{-\epsilon n}{4\lambda}\lambda \epsilon\bigg)=\exp\bigg(\frac{-\epsilon^2 n}{4}\bigg),
\end{gather*}
and combining these bounds gives
\begin{equation*}
    \Pr\bigg( \bigg\vert \jgedit{Z_j}-\frac{j}{\lambda}\bigg\vert >\frac{\epsilon n}{\lambda}\bigg)\leq 2\exp\bigg(\frac{-\epsilon^2 n}{4}\bigg)
\end{equation*}
Now apply Lemma \ref{lem:etemadi} using the centered random variables $Y_i=X_i-1/\lambda$. This yields
\begin{align*}
    \Pr\bigg(\max_{1\leq j\leq n} \bigg\vert \jgedit{Z_j}-\frac{j}{\lambda}\bigg\vert > \frac{\epsilon n}{\lambda}\bigg)&\leq 3\max_{1\leq j\leq n}\Pr\bigg( \bigg\vert \jgedit{Z_j}-\frac{j}{\lambda}\bigg\vert > \frac{\epsilon n}{3\lambda}\bigg)\\
    &\leq 6  \exp\bigg(\frac{-\epsilon^2n}{36}\bigg).
\end{align*}
\end{proof}

\begin{corollary}
\label{cor:geom3}
Let $\{G_{i,j}\}_{i\in [n],j\in [w]}$ be a family of independent geometric random variables such that for all $i,j$,
\begin{equation*}
    G_{i,j}\sim \text{Geom}(\lambda_i).
\end{equation*}
Let $\jgedit{Z_k^i}=\sum_{j=1}^k G_{i,j}$. Then for any $\epsilon\in [0,1]$, 
\begin{equation}
\label{eq:geombound}
    \Pr\bigg(\exists i\in [n],j\in [w]: \bigg\vert \jgedit{Z_k^i}-\frac{k}{\lambda_i}\bigg\vert \geq \frac{\epsilon w}{\lambda_i}\bigg)\leq 6n\exp\bigg(\frac{-\epsilon^2 w}{36}\bigg).
\end{equation}
\end{corollary}
\begin{proof}
This follows immediately from Corollary \ref{cor:geom2} and a union bound.
\end{proof}

\begin{lem}
\label{lem:bern}
Let $\{\jgedit{S}_{i,j}\}_{i\in \jgedit{[m]}, j\in [w]}$ be an independent Bernoulli ensemble such that for all $i,j$
\begin{equation*}
    \jgedit{S}_{i,j}\sim \text{Bern}(\mu_i),
\end{equation*}
with $\mu_1\geq \mu_2\geq\ldots\geq \mu_n$.
Then for all $\delta\in [0,1]$,
\begin{equation}
\label{eq:bernbound}
    \Pr\bigg(\exists k\in \jgedit{[m]}:\sum_{i=1}^k\sum_{j=1}^w \jgedit{S}_{i,j}\leq (1-\delta)w\bigg(\sum_{i=1}^k \mu_i\bigg)\bigg)\leq \jgedit{m}\exp\bigg(\frac{-\delta^2 w\mu_1}{2}\bigg)
\end{equation}
\end{lem}
\begin{proof}
The well-known multiplicative form of the Chernoff bound immediately implies that for each $k\in \jgedit{[m]}$,
\begin{equation*}
    \Pr\bigg(\sum_{i=1}^k\sum_{j=1}^w \jgedit{S}_{i,j}\leq (1-\delta)w\bigg(\sum_{i=1}^k \mu_i\bigg)\bigg)\leq \exp\bigg(\frac{-\delta^2 w\sum_{i=1}^k \mu_i}{2}\bigg)\leq \exp\bigg(\frac{-\delta^2 w\mu_1}{2}\bigg).
\end{equation*}
The result then follows from a union bound over all $k\in \jgedit{[m]}$.
\end{proof}

The following characterizes the moments of geometric distributions.
\begin{lem}
\label{lem:geomom}
Let $X\sim \text{Geom}(\lambda)$. Then for all $k\geq 1$
\begin{equation*}
    \mathbb{E}[X^k]\leq \frac{c_k}{\lambda^k},
\end{equation*}
where $c_k$ is a constant depending on $k$ but not on $\lambda$.
\end{lem}

\begin{lem}
\label{lem:binbound}
Let $X\sim \text{Bin}(n,p)$, where $p\in (0,1]$ is considered fixed. Then, for any fixed integer $k\geq 0$,
\begin{equation}
    \mathbb{E}[X^k]\asymp n^k,
\end{equation}
where the implicit constants depend on $p$ and $k$, but not $n$.
\end{lem}
\begin{proof}
By definition, $X=\sum_{i=1}^n X_i$, where $X_i\sim \text{Bern}(p)$ are i.i.d. We clearly have
\begin{equation}
    X^k=\sum_{1\leq i_1,\ldots,i_k\leq n} \prod_{j=1}^k X_{i_j}.
\end{equation}
Note that products of these indicator variables remain indicator random variables, and it is easy to see that for any indices $1\leq i_1,\ldots,i_k\leq n$,
\begin{equation}
    p^k\leq \mathbb{E}[\prod_{j=1}^k X_{i_j}]\leq p.
\end{equation}
Therefore, taking expectations and summing we obtain
\begin{equation}
    p^k n^k\leq \mathbb{E}[X^k]\leq pn^k,
\end{equation}
as desired.
\end{proof}

\subsection{Proofs for Section 2}
\subsubsection{Central Feasibility}
We will need the following results and definitions:
\begin{defn}
Let $\mathbf{x},\mathbf{y}\in \mathbb{R}^n_+$, and assume that $x_1\geq x_2\geq \ldots\geq x_n\geq 0$ and analogously for $y$. Then $\mathbf{x}$ \textbf{weakly dominates} $\mathbf{y}$ if for each $1\leq k\leq n$ 
\begin{equation*}
    \sum_{i=1}^k x_i\geq \sum_{i=1}^k y_i.
\end{equation*}
If the above inequalities are strict for each $1\leq k\leq n$, then $\mathbf{x}$ \textbf{strictly dominates} $y$. If $\mathbf{x}$ weakly dominates $\mathbf{y}$, and further
\begin{equation*}
    \sum_{i=1}^n x_i=\sum_{i=1}^n y_i,
\end{equation*}
then $\mathbf{x}$ is said to \textbf{majorize} $\mathbf{y}$. If the dimensions disagree, one can extend this definition in the natural way by padding the shorter vector with zeros.\footnote{Weak dominance is usually referred to as \textbf{weak majorization}; we change the terminology slightly as strict domination is the relevant property in our setting.}
\end{defn}
\begin{defn}
A nonnegative square matrix $P\in \mathbb{R}^{n\times n}$ is \textbf{doubly stochastic} if each row and column sums to $1$.
\end{defn}
We need the following facts about dominance:

\jgedit{\begin{lem}
\label{lem:dominancelem}
Suppose $\mathbf{x}$ weakly dominates $\mathbf{y}$. Then for any nonnegative, monotone decreasing sequence $z_1\geq \ldots\geq z_n\geq 0$,
\begin{equation}
    \sum_{i=1}^n z_i x_i\geq \sum_{i=1}^n z_i y_i
\end{equation}
\end{lem}
\begin{proof}
Simply multiply the equations by appropriate scalars in the definition of weak dominance and sum to get the desired inequality.
\end{proof}
}
\begin{lem}[Theorem B.2. in \cite{marshall1979inequalities}]
Suppose $\mathbf{x},\mathbf{y}\in \mathbb{R}^n_+$ are in sorted order, and $\mathbf{x}$ majorizes $\mathbf{y}$. Then $\mathbf{y}=P\mathbf{x}$ for some doubly stochastic matrix $P$.
\end{lem}

\begin{corollary}
\label{cor:dom}
Suppose $\mathbf{x},\mathbf{y}\in \mathbb{R}^n_+$ are in sorted order and $x$ strictly dominates $y$. Then there exists a doubly stochastic matrix $P$ such that $P\mathbf{x}$ is strictly greater than $\mathbf{y}$ componentwise.
\end{corollary}
\begin{proof}
By continuity and strict dominance, it is possible to scale all entries of $\mathbf{x}$ by nonnegative factors strictly less than $1$ to obtain a vector $\mathbf{x'}$ that majorizes $\mathbf{y}$. Applying the previous result, we have $P\mathbf{x'}=\mathbf{y}$ for some doubly stochastic $P$. But $P\mathbf{x}$ strictly exceeds $P\mathbf{x'}$ componentwise, giving the result.
\end{proof}

We can now proceed with the proof of the Theorem \ref{thm:feasibility}:

\begin{thm}[Thm. \ref{thm:feasibility}, restated]
Suppose $\bm{\mu}$ and $\bm{\lambda}$ have been preprocessed so that a maximal, equal prefix of $1$'s is deleted from both and are nonempty and not both identically zero afterwards.\footnote{This assumption is without loss of generality; it is easy to see that one can always match these queues and servers in every round, and then the stability of the entire system is dictated by the rest of the queues and servers. Moreover, this can be assumed without loss of generality, as any scheduling strategy that does not match these queues and servers infinitely often clearly will have unbounded buildup, violating strong stability.} Then the above queuing system is strongly stable for some centralized (coordinated) scheduling policy if and only if for all $1\leq k\leq n$,
\begin{equation}
    \sum_{i=1}^k \mu_i > \sum_{i=1}^k \lambda_i.
\end{equation}
\end{thm}
\begin{proof}
\textbf{Sufficiency}: First suppose $\bm{\mu}$ strictly dominates $\bm{\lambda}$, when appropriately appending zeros if needed to make the vectors of same length. By Corollary \ref{cor:dom}, there exists some doubly stochastic $P$ such that $P\bm{\mu}>\bm{\lambda}$. Moreover, by the well-known Birkhoff-von Neumann Theorem, the set of doubly stochastic matrices is the convex hull of the set of permutation matrices $\mathcal{P}$. This implies there exists a distribution $\pi$ over $\mathcal{P}$ such that $\bm{\lambda}< P\bm{\mu}$, interpreted componentwise, where $P=\mathbb{E}_{\Pi\sim \pi}[\Pi]$.

Consider the following oblivious scheduling algorithm: at each time $t$, independently sample a permutation matrix $\Pi$ from $\pi$, and schedule queues via the associated matching on the bipartite graph of queues and servers (even if some queues have no available packets to send). For each queue $i$, the associated marginal distribution on servers it sends to in each round is given by the $i$th row of $P$. Given that queue $i$ has a packet to send at time $t$, the probability of successfully clearing a packet is exactly $(P\bm{\mu})_i>\lambda_i$, as this scheduling scheme ensures each queue is alone at each server it sends to. As a result, the packet clears so long as the server is successful. Therefore, the random process $Q^i_t$ of number of packets by queue $i$ at time $t$ follows a homogeneous random walk on the half-line biased towards $0$, which is ergodic with a stationary distribution with geometric tails. It is not difficult to show that any distribution on the natural numbers with geometric tails has bounded $r$th moments for any $r\geq 0$.\footnote{This can also easily be seen directly using Theorem \ref{thm:pemantle}. Negative drift when exceeding $Q_t=0$ is obvious, and as queue sizes can change by at most $n$ in total between steps, increments are clearly bounded in $L^p$ for any $p\geq0$.} This then extends to the $r$th moment of the sum by Minkowski's inequality, as the $L^r$ norm of random variables satisfies the Triangle Inequality. This proves strong stability when strict dominance holds.

\textbf{Necessity}: It suffices to show that if one of the above inequalities fails, the first moment of $Q_t$ is unbounded over time. To that end, first suppose that strict dominance is strictly violated, namely there is some $k\leq n$ such that $\sum_{i=1}^k \lambda_i>\sum_{i=1}^k \mu_i$. Let $Q^{\leq k}_t=\sum_{i=1}^k Q^i_t$ be the total number of packets at the $k$ queues with highest arrival rate. Under any scheduling policy, the difference between $Q_{t+1}^{\leq k}$ and $Q_t^{\leq k}$ is bounded below in expectation by $\sum_{i=1}^{k}\lambda_i-\sum_{i=1}^k \mu_i>0$, as $\sum_{i=1}^k \lambda_i$ new packets arrive for these queues at each step in expectation, and at most $\sum_{i=1}^k \mu_i$ packets can be cleared in expectation. In particular, as $Q_t:=\sum_{i=1}^n Q_t^i\geq Q_t^{\leq k}$ surely by nonnegativity of queue sizes, telescoping gives
\begin{equation*}
    \mathbb{E}[Q_t]\geq \mathbb{E}[Q_t^{\leq k}]=\sum_{s=0}^{t-1}\mathbb{E}[Q_{s+1}^{\leq k}-Q_{s}^{\leq k}]\geq t(\sum_{i=1}^k \lambda_i-\sum_{i=1}^k \mu_i)\to \infty.
\end{equation*}

To extend this to when strict dominance is only weakly violated, namely there is some $k\leq n$ such that $\sum_{i=1}^k \lambda_i=\sum_{i=1}^k \mu_i$, we will need one more tool. Again, it is sufficient to show that $\mathbb{E}[Q^{\leq k}_t]\to \infty$. The previous argument actually shows that $Q^{\leq k}_t$ is a nonnegative submartingale for any measurable scheduling policy. If $\lim_{t\to \infty}\mathbb{E}[Q^{\leq k}_t]= \sup_t \mathbb{E}[Q^{\leq k}_t]<\infty$, then the Martingale Convergence Theorem (Theorem 4.2.11 of \cite{durrett2019probability}) implies that there exists an almost surely finite random variable $Q^{\leq k}_{\infty}$ such that $\lim_{t\to\infty} Q_t^{\leq k}\to Q_{\infty}^{\leq k}$ almost surely. But $Q_{t+1}^{\leq k}-Q_t^{\leq k}$ is integer-valued and not equal to zero with nonzero probability unless $\bm{\mu}$ and $\bm{\lambda}$ are degenerate in the sense that all entries are $0$ or $1$, but this is ruled out by the assumption. This implies the pointwise limit cannot exist unless the limit is infinite, but this violates the almost sure finiteness of $Q_{\infty}^{\leq k}$, a contradiction.
\end{proof}

\subsubsection{Impossibility for No-Priority Model}
Next, we give the promised example that the simpler queuing model is too weak to give any sub-polynomial bicriterion result:

\begin{thm}[Theorem \ref{thm:impossibility}, restated]
In the alternate model, for large enough $n$, there exists a centrally feasible queuing system with $n$ queues and servers with the following property: \jgedit{the system remains feasible even if $\bm{\lambda}$ is scaled up by $\Omega(n^{1/3})$ and} it is possible for all queues to be in a Nash equilibrium at each time step (and in particular, satisfy no-regret properties as in Assumption \ref{assumption:learning}), yet the system is not strongly stable.
\end{thm}
\begin{proof}
Let $\lambda_1=2/n^{1/3}$, while $\lambda_2=\ldots=\lambda_n=1/n^{2/3}$; let $\mu_1=1/2$ and $\mu_2=\ldots=\mu_n=c/n^{1/3}$, where $c=c(n)=\Theta(1)$ is such that 
\begin{equation*}
    \frac{1}{n^{1/3}+2}<\frac{c}{n^{1/3}}<\frac{1}{n^{1/3}}.
\end{equation*}

We proceed by considering an adversarial, centralized scheduler that suggests actions for each queue in each round, while enforcing that each agent achieves no regret (even further, each round is a Nash equilibrium). The schedule is as follows: in each round, the scheduler chooses $n^{1/3}/2-1$ of the low rate agents arbitrarily to send to the unique high rate server, if that many low rate agents have packets, as well as the high rate queue. All other low rate agents send to distinct low rate servers. If fewer that $n^{1/3}/2-1$ low rate servers are active, then the scheduler schedules all active queues to the high rate server.

By standard Chernoff bounds, the number of low rate queues that receive a packet in a given round is at least $n^{1/3}/2-1$ with probability at least $1-\exp(-\Omega(n^{1/3}))$, so with at least this probability there are enough low rate queues for the first case to hold. The inequalities above show that in such a round where there are at least $n^{1/3}/2-1$ active low agents, the suggested schedule is a Nash equilibrium, and the probability of success for each queue sending to the high server is exactly $1/n^{1/3}$ in such rounds. When this does not occur, the suggested schedule is still Nash, and the probability of success for any queue sending to the high rate server is at most $1/2$. Therefore, in any time step where the high rate queue has a packet, by the Law of Total Probability, her probability of clearing is upper bounded by
\begin{equation*}
    \frac{1}{n^{1/3}}+\exp(-\Omega(n^{1/3}))\cdot (1/2)< \frac{1.5}{n^{1/3}}
\end{equation*}
where the inequality is for sufficiently large $n$. As a result, in expectation $Q^1_{t+1}-Q^1_t$ is lower bounded by a nonzero constant (depending on $n$, but not on $t$), and therefore $Q^1_t$ diverges with $t$ in expectation by telescoping. This shows that this system is not strongly stable, even though every queue plays a Nash strategy at each time. Note that this system would still be centrally feasible if all queues were scaled up by a factor of $\Theta(n^{1/3})$, giving the result.

To see that this is no-regret with high probability on each fixed window, define $S^{i,j}_t$ to be the indicator variable that queue $i$ \emph{would} succeed in clearing a packet at server $j$ at time $t$, and let $\sigma_i(t)$ be the identity of the server that queue $i$ chooses at time $t$. Note that if queue $i$ is empty at time $t$, then $S^{i,j}(t)=0$ for all $j$ and $\sigma_i(t)$ can be arbitrary. Then, define $\Delta^{i,j}_t=S^{i,\sigma_i(t)}_t-S^{i,j}_t$. By the Nash discussion above, $\mathbb{E}[\Delta^{i,j}_t\vert \mathcal{F}_{t-1}]\geq 0$ for all $t$ in both cases as described above, where $\mathcal{F}_{t}$ denotes the past history of this process up to time $t$. This holds regardless of if queue $i$ is really sending in that round (in which case the quantity is just $0$).

Therefore, as $\vert \Delta^{i,j}_t\vert\leq 2$ surely, we may apply the Azuma-Hoeffding inequality (Lemma \ref{lem:azuma}) to see that on any fixed window of length $w$ (and reindexing time so that time progresses $t=1,\ldots,w$ on this window for notational ease)
\begin{equation}
    \Pr\bigg( \sum_{t=1}^w \Delta^{i,j}_t\leq -\alpha\bigg)\leq \exp\bigg(\frac{-\alpha^2}{w}\bigg).
\end{equation}
By a union bound, for each queue $i$, this holds for all servers $j\in [m]$ with probability at most $m\cdot \exp\bigg(\frac{-\alpha^2}{w}\bigg)$. Note that if $\alpha=\sqrt{w\ln(m/\delta)}$, this quantity is at most $\delta$. As such, by definition of regret, on any fixed period of length $w$, with probability at least $1-\delta$, this strategy satisfies
\begin{equation}
    \text{Reg}_i(w)\leq \sqrt{w\ln(m/\delta)}=o(w),
\end{equation}
as needed.
\end{proof}

\subsection{Proofs for Section 3}

\subsubsection{Relationship Between Forms of Stability}
We can now show the desired relations between strong stability and almost sure stability. We need the following technical lemma:
\begin{lem}
\label{lem:stabimplication}
Suppose a nonnegative sequence of random variables $X_1,X_2,\ldots$ satisfies $X_t\leq X_{t-1}+L$ surely for some fixed $L\geq 0$ and any $t$, as well as the uniform moment condition $\sup_{t}\mathbb{E}[X_t^p]\leq C_p$ for some constant $C_p\geq 0$ for all $p\geq 1$. Then, for any $c>0$, almost surely, $X_t=o(t^c)$.
\end{lem}
\begin{proof}
Fix $\epsilon>0$. It suffices to prove the lemma for $0<c<1$, so take $0<d<c$ and set $p=d^{-1}$. We do this by proving the desired asymptotics on a conveniently chosen subsequence, then interpolate to intermediate values. Indeed, by Markov's inequality, for each $k\geq 1$
\begin{equation*}
    \Pr(X_{k^{1+\epsilon}}> k^{(1+\epsilon)d})= \Pr(X_{k^{1+\epsilon}}^p> k^{1+\epsilon})\leq \frac{C_p}{k^{1+\epsilon}}.
\end{equation*}
Summing over $k$ and observing the right side is summable, we deduce from the first Borel-Cantelli Lemma that almost surely, for all sufficiently large $k$,
\begin{equation*}
    X_{k^{1+\epsilon}}\leq k^{(1+\epsilon)d}.
\end{equation*}
To extend this to all large enough $t$, suppose that $t$ is such that $k^{1+\epsilon}\leq t< (k+1)^{1+\epsilon}$. By the one-sided boundedness, we know that almost surely, for such $t$ and all large enough $k$,
\begin{align*}
    X_t&\leq L\cdot (t-k^{1+\epsilon})+X_{k^{1+\epsilon}}\\
    &\leq L\cdot (1+\epsilon)(k+1)^{\epsilon} + k^{(1+\epsilon)d}\\
    &\leq L\cdot(1+\epsilon)(t^{1/(1+\epsilon)}+1)^{\epsilon}+t^d,
\end{align*}
where the bound on $t-k^{1+\epsilon}$ arises from the Mean Value Theorem. Clearly this last expression is $O(t^{\epsilon/(1+\epsilon)}+t^d)$. As this holds for arbitrary $\epsilon>0$, we may take $\epsilon$ small enough so that this expression is $o(t^c)$, as claimed.
\end{proof}

\begin{remark}
We have shown that for random process satisfying  the conditions of the previous lemma grows at most subpolynomially. It is perhaps interesting to find a corresponding lower bound: it is possible that almost surely, such a process exceeds $\Omega(\sqrt{\ln\ln t})$ infinitely often. This can be seen by considering the scaled simple random walk $\vert S_t\vert/\sqrt{t}$ on the integers. It is well-known that $\mathbb{E}[(\vert S_t\vert/\sqrt{t})^p]\leq C_p$ for some constant $C_p$ depending only on $p$ via the Central Limit Theorem, and yet by the Law of the Iterated Logarithm, $\limsup_{t\to \infty} \vert S_t\vert/\sqrt{t\ln\ln t}=\sqrt{2}$ almost surely (Theorem 9.5 of \cite{billingsley2008probability}).
\end{remark}

\begin{remark}
Clearly, if a nonnegative random process satisfies $\mathbb{E}[X_t^p]\leq C_p$, one can derive simple bounds on the probability that $X_t$ exceeds any given threshold $\lambda$ just via Markov's inequality. The result above leverages moment control to give asymptotic bounds.
\end{remark}

\begin{lem}[Lemma \ref{lem:equivalence}, restated]
If the queuing dynamics are such that queues select servers independently of any information about new received packets after their current oldest packet was received, then strong stability in the original system is equivalent to strong stability in the dual system. 

As a corollary, if the standard system and the duual system are equivalent processes, then strong stability in either system implies almost sure stability.
\end{lem}
\begin{proof}
Suppose that the dynamics are as stated, so that the standard and dual dynamics yield completely equivalent processes. Then the distribution of $Q^i_t$ conditioned on the value of $T^i_t$ at time $t$ is $\text{Bin}(T^i_t,\lambda_i)$. Note that by the Law of Iterated Expectations, $\mathbb{E}[(Q^i_t)^p]=\mathbb{E}[\mathbb{E}[(Q^i_t)^p\vert T^i_t]]$. But by Lemma \ref{lem:binbound}. $\mathbb{E}[(Q^i_t)^p\vert T^i_t]\asymp (T_i^t)^p$ up to absolute constants depending only on $p$ and $\lambda_i$.
Therefore, by taking expectations, the standard system and dual system have equivalent strong stability properties.

Almost sure stability now follows from either form of strong stability from Lemma \ref{lem:stabimplication}, noting that either $Q_t$ or $T_t$ can increase by at most $L=n$ in each time step.

\end{proof}

\end{document}

%% file: abstract.tex
Bounding the price of anarchy, which quantifies the damage to social welfare due to selfish behavior of the participants, has been an important area of research. In this paper, we study this phenomenon in the context of a game modeling queuing systems: routers compete for servers, where packets that do not get service will be resent at future rounds, resulting in a system where the number of packets at each round depends on the success of the routers in the previous rounds. We model this as an (infinitely) repeated game, where the system holds a state (number of packets held by each queue) that arises from the results of the previous round. We assume that routers satisfy the no-regret condition, e.g. they use learning strategies to identify the server where their packets get the best service. 

Classical work on repeated games makes the strong assumption that the subsequent rounds of the repeated games are independent (beyond the influence on learning from past history). The carryover effect caused by packets remaining in this system makes learning in our context result in a highly dependent random process. We analyze this random process and find that 
if the capacity of the servers is high enough to allow a centralized and knowledgeable scheduler to get all packets served even with double the packet arrival rate, and queues use no-regret learning algorithms, then the expected number of packets in the queues will remain bounded throughout time, assuming older packets have priority. This paper is the first to study the effect of selfish learning in a queuing system, where the learners compete for resources, but rounds are not all independent: the number of packets to be routed at each round depends on the success of the routers in the previous rounds.

%% file: intro.tex
\section{Introduction}
In this paper, we consider how to guarantee the efficiency of stochastic queuing systems when routers use simple learning strategies to choose the servers they use, and repeatedly resend their packets until the packet gets served. We show conditions that guarantee the stability of such systems despite the competition of queues and the carryover effects between rounds caused by resending packets.

Understanding how to design complex systems that remain efficient even when used by selfish agents is an important goal of algorithmic game theory. The \emph{price of anarchy} \cite{Koutsoupias2009Worst} measures this inefficiency by comparing the welfare of the Nash equilibrium of the game to the socially optimal solution without considering incentives. This notion has lead to a long line of literature bounding \jgedit{this loss in} various games. Our results are analogous in spirit to those of \cite{RoughgardenT2002}, which shows that in the context of routing in networks with delay, the cost of any Nash equilibrium outcome is no more expensive than the centrally designed optimum that carries twice as much flow. 

We model the behavior of queues as learners, assuming that their choices of where to send packets satisfies the no-regret guarantee. This guarantee can be ensured by running any of a large set of learning algorithms \cite{book}. Studying learning behavior in games has a long history, dating back to the early work of Robinson \cite{Robinson}, see also \cite{FudenbergLevine}. 
Work in the last two decades has extended the Nash equilibrium quality analysis (price of anarchy) to learning outcomes \cite{BlumHLR,Roughgarden09,SyrgkanisT13}. If all players employ a no-regret learning strategy, then the play converges to a form of correlated equilibrium of the game \cite{Hart} (players correlating their play by each of them using the history of play to decide their next action), and the price of anarchy analysis extends also to the correlated play.  A serious limitation of the model of repeated games considered in these works is the assumption that the games played at different times are independent, in the sense that the outcome of the game at time $t$ has no direct  effect at time $t+1$ except through the learning of the agents. While this can be a good approximation in some games, there are many applications where this clearly fails. In the context of routing games modeling the morning rush-hour traffic, it is indeed the case that no matter how bad Monday morning traffic was, the traffic jam clears up before Tuesday morning. In contrast, when modeling traffic on a smaller time scale, if one car or packet experiences delay, it remains in the system longer, or may be re-sent later, and hence affects later time periods. In ad-auctions, the remaining budget of the player has a similar carryover effect, as 
winning the auction in one round decreases the player's remaining budget in future rounds.

Here, we will consider a queuing system with queues sending packets to servers as a simplified model of a network of queues. 
In  \cite{krishnasamy2016learning}, the authors study the performance of a learning algorithm in the same queuing system finding the best server(s) with respect to \emph{queue-regret}, which measures the expected difference in queue sizes to that of a genie strategy that knows the optimal server. Their primary goal is to study this more refined notion for the queuing setting for this classical stochastic bandit problem, which exhibits more complicated behavior than standard no-regret bounds that grow at least logarithmically with time. 
Our paper studies a decentralized, multi-queue version of the same system, where each queue uses their own learning algorithm to clear their packets while selfishly competing with each other for service. Our primary focus is on establishing precisely when standard no-regret algorithms can ensure that the competitive queuing system remains stable even in this game-theoretic setting, a concern that does not arise 
in the learning problem with centralized scheduling.

\headder{Our Results.}
Our main result concerns a multi-agent version of the queuing system of \cite{krishnasamy2016learning}, where the queues each use their own no-regret learning algorithm to find and compete for the best servers. We show that if the service rates of the servers is high enough to allow a centralized scheduler to get all packets served even with double the arrival rate and when older packets have priority over younger packets, then the expected length of all queues will remain bounded 
over all time, 
 assuming the learning algorithms used satisfy the no-regret assumption. Studying the outcome of learning in such systems with carryover effect requires us to study these interactions not just as a repeated game, but as a highly dependent random process. 

In this model, a server can serve at most one packet at any time, and packets remaining in the system are queued at the input side. At each time step, a server with service rate $\mu$ will select one of the packets sent to it (if any), serve it with probability $\mu$, and return all unserviced packets to their queues. In other queuing systems, the servers may also have a bounded size queue and would only send back (or drop) packets when they no longer fit on the queue; our simpler model without server queues makes the trade-offs we want to study cleaner. A packet sent to a server is either served or returned and offers instantaneous feedback to the learning algorithms of the queues, in contrast to the bit more informative, but delayed feedback available in real systems. 

An important feature of the model is how conflicts are resolved when multiple queues send to the same server in a 
time period.
We show that if the servers select a packet uniformly at random among the arriving packets, then unless the success rates of the servers are prohibitively larger than the arrival rates of the queues (by a factor that grows with the number of queues), learning does \emph{not} necessarily ensure that all packets will get served in a timely manner: in systems with many queues with low arrival rate, when these queues are selfishly aiming to get good service, the number of unserved packets at queues 
with high arrival rate may grow linearly over time. Our main result is to show that if packets also carry a timestamp, and servers choose to serve the oldest arriving packet, then this linear blowup cannot arise: if the system has enough capacity to serve all packets when they are centrally coordinated, even with just double the arrival rate, then no-regret learning of the queues guarantees that all packets get served and queue lengths stay bounded \jgedit{in expectation}. We also show that this bound of 2 on the required service rate is tight, in that with less than a factor of 2 higher service rate, no-regret learning does not necessarily guarantee the timely service of all packets.

\headder{Our Techniques.}
The carryover effect between rounds caused by packets left in the system, forces us to study these interactions not just as a repeated game, but as a highly dependent random process. Moreover, the randomness arises intrinsically from both the randomized strategic behaviors of the queues and the inherent randomness in the queuing system. To establish the result, we combine game-theoretic properties implied by the no-regret assumption with techniques from random processes to establish the high-probability results.

In analysing the behavior of the queuing system, we have to deal with highly dependent processes. If a queue receives too many packets during a previous period, this has a major effect not only on the outcomes for this queue, but for every other queue it may be competing with.  
To make the study of these random process more manageable, we use the principle of deferred decisions: rather than considering the state of the queue sizes, each with possibly many packets, we keep track only of the timestamp of the oldest packet in each queue and defer seeing when the next packet arrived until after this one is served. In doing so, the timestamp of the next packets to be cleared and the service successes of the servers are all independent of the current time period, and hence we can use standard concentration bounds. 

To prove the bounds on the queue sizes, we use a potential function based on the oldest time stamp in each queue. The main idea of the proof is to argue that when this function has a high enough value, than it must have negative drift. To conclude that the queues remain bounded, we use a powerful theorem of Pemantle and Rosenthal \cite{pemantle1999moment} showing that a sufficiently regular stochastic process with negative drift must have moments uniformly bounded over time.
Once we obtain this property for our random process, we then use standard probabilistic techniques to obtain an evidently weaker, but perhaps more interesting, asymptotic control on the almost sure growth of the queues in these queuing systems. We hope that the kinds of qualitative features we establish and the methods of proof for these results will be of interest in the future study of repeated strategic interactions that similarly relax the independence assumptions of the games played at each round.

\headder{Further Related Work.}
As already explained above, the model we study combined features of learning in games with classical queuing systems; both of these areas have large bodies of literature. The classical focus of work on scheduling in queuing systems is to identify policies that achieve optimal throughput (see for example the textbook of \cite{queuing_theory}). Closest to our model from this literature is the work of \cite{krishnasamy2016learning}, which characterizes the queue-regret of learning algorithms that only seek to identify the best servers, but does not consider competition between selfish learners. They characterize queue-regret for the case of a single queue aiming to find the best server, and extend the result to the case of multiple queues scheduled by a single coordinated scheduling algorithm, assuming there is a perfect matching between queues and optimal servers that can serve them. In contrast, we assume that each queue separately learns to selfishly make sure its own packets are served at a high enough rate, offering a game-theoretic model of scheduling packets in a queuing system, and do not make the matching assumption on queues and servers. Compared to classical price of anarchy bounds in repeated games \cite{BlumHLR,Roughgarden09,SyrgkanisT13}, we no longer make the assumption that games at different rounds are independent. Studying this model requires us to combine ideas from the price of anarchy analysis of games with understanding the behavior of stochastic systems.

Our work is one of the first examples of studying the effect of learning in games with carryover effects between rounds. Studying such systems requires understanding a highly  dependent random process. Among the large body of literature of such processes, closest to our work is the adversarial queuing systems of  \cite{adversarial_queuing}, who also use the Pemantle and Rosenthal \cite{pemantle1999moment} theorem to establish bounded queue sizes in expectation. 
Another important 
repeated game setting with such carryover effect is the repeated ad-auction game with limited budgets. The papers of \cite{nikhil_pacing, pacing, pacingEC} consider such games and offers results on convergence to equilibrium as well as understanding equilibria in the first-price auction settings under a particular behavioral model of the agents. Analyzing such systems for the more commonly used second-price auction system is an important open problem.